\documentclass[11pt]{article}
\usepackage[margin=1in]{geometry}
\usepackage[ruled]{algorithm}
\usepackage{svg}
\usepackage{algpseudocode, xcolor, amsmath, mathtools, amssymb, amsthm}
\usepackage{thm-restate}
\usepackage{soul}
\usepackage{pgfplots}
\usepackage{hyperref}
\usepackage[capitalize]{cleveref}
\newcommand{\mbf}{\mathrm{MBF}}
\newcommand{\Bf}{\mathrm{BF}}
\newcommand{\Nf}{\mathrm{NF}}
\newcommand{\opttt}{\mathtt{Opt}}
\newcommand{\algtt}{\mathtt{Alg}}
\newcommand{\opt}{\mathrm{Opt}}
\newcommand{\Opt}{\mathrm{Opt}}
\newcommand{\optv}{\mathrm{Opt}_v}
\newcommand{\Th}{^{\mathrm{th}}}
\newcommand{\floor}[1]{\left\lfloor#1\right\rfloor}
\newcommand{\ceil}[1]{\left\lceil#1\right\rceil}
\newcommand{\abs}[1]{\left|#1\right|}
\newcommand{\prob}[2][]{\mathbb{P}_{#1}\left[#2\right]}
\newcommand{\expec}[2][]{\mathbb{E}_{#1}\left[#2\right]}
\newcommand{\var}[2][]{\mathrm{Var}_{#1}\left[#2\right]}
\newcommand{\cov}[3][]{\mathrm{Cov}_{#1}\left[#2,#3\right]}

\newcommand{\Isamp}{I_{\mathrm{samp}}}

\newcommand{\sample}{\delta^2n}

\newcommand{\param}{\delta}
\newcommand{\weight}[1]{\mathcal{W}(#1)}
\newcommand{\vol}[1]{\mathrm{vol}(#1)}
\newcommand{\stage}[1]{T_{#1}}
\newcommand{\stagelarge}[1]{L_{#1}}
\newcommand{\stagesmall}[1]{S_{#1}}

\newcommand{\auxalgo}{\mathcal A_\alpha}
\newcommand{\largeitems}{I_{\ell}}
\newcommand{\smallitems}{I_{s}}
\newcommand{\guess}[1]{n_{#1}}
\newcommand{\superstage}[1]{\Gamma_{#1}}

\newcommand{\unmatch}{\mathrm{unmatch}}

\newcommand{\mini}{\mathrm{small}}
\newcommand{\bluep}{\mathtt{BlueP}}
\newcommand{\rank}{\mathrm{rank}}
\newcommand{\main}{\mathrm{main}}
\newcommand{\of}{\mathrm{of}}

\newcommand{\bestfit}{Best-Fit}
\newcommand{\nextfit}{Next-Fit}
\newcommand{\iid}{i.i.d.}
\newcommand{\whp}{w.h.p.}

\renewcommand{\epsilon}{\varepsilon}

\newcommand{\ALG}{\mathtt{Alg}}
\newcommand{\IMPALG}{\mathtt{ImpAlg}}
\newcommand{\ALGV}{\mathtt{Alg}_v}
\newcommand{\OPT}{\mathrm{Opt}}
\newcommand{\eps}{\varepsilon}

\numberwithin{equation}{section}

\makeatletter
\def\blfootnote{\gdef\@thefnmark{}\@footnotetext}
\makeatother

\newtheorem{theorem}{Theorem}[section]
\newtheorem{lemma}{Lemma}[section]
\newtheorem{claim}{Claim}[section]
\newtheorem{remark}{Remark}[section]
\newtheorem{corollary}{Corollary}[section]
\newtheorem{proposition}{Proposition}[section]

\DeclareMathOperator*{\argmax}{\arg\max}
\DeclareMathOperator*{\argmin}{\arg\min}

\hypersetup{
    colorlinks=true,
    urlcolor=magenta,
    citecolor=magenta
}

\title{Near-optimal Algorithms for Stochastic Online Bin Packing\footnote{A preliminary version of this work appeared in the 49th EATCS International Colloquium on Automata, Languages and Programming (ICALP), 2022. }} 

\author{Nikhil Ayyadevara\footnote{University of Michigan, Ann Arbor, USA, \texttt{vsnikhil@umich.edu}}
\and Rajni Dabas\footnote{Northwestern University, Evanston, USA, \texttt{rajni.dabas1@northwestern.edu}\\ \null$\quad\:\:\;$University of Delhi, New Delhi, India, \texttt{rajni@cs.du.ac.in}}
\and Arindam Khan\footnote{Indian Institute of Science, Bengaluru, India, \texttt{arindamkhan@iisc.ac.in}}
\and K. V. N. Sreenivas\footnote{Indian Institute of Science, Bengaluru, India, \texttt{venkatanaga@iisc.ac.in}}
}

\date{\empty}

\newcommand{\kvnr}[1]{}
\newcommand{\nar}[1]{}
\newcommand{\arir}[1]{}
\newcommand{\rdr}[1]{}

\newcommand{\kvn}[1]{{#1}}

\begin{document}
\maketitle
\begin{abstract}
We study the online bin packing problem under two stochastic settings. In the bin packing problem, we are given $n$ items with sizes in $(0,1]$ and the goal is to pack them into the minimum number of unit-sized bins.

First, we study bin packing under the  i.i.d.~model, where item sizes are sampled independently and identically from a distribution in $(0,1]$.
Both the distribution and the total number of items are unknown. 
The items arrive one by one and their sizes are revealed upon their arrival and they must be packed immediately and irrevocably
in bins of size $1$. 
We provide a simple meta-algorithm that takes an offline $\alpha$-asymptotic approximation algorithm and
provides a polynomial-time $(\alpha + \eps)$-competitive algorithm for online bin packing under the  i.i.d.~model, where $\eps>0$ is a small constant.
Using the AFPTAS for offline bin packing, we thus provide a linear time $(1+\eps)$-competitive algorithm for online bin packing under i.i.d.~model, thus settling the problem.

We then study the random-order model, where an adversary chooses the instance, 
but the order of arrival of items in the instance is drawn uniformly
at random from the set of all permutations of the items. Kenyon's seminal result [SODA '96] showed that the Best-Fit algorithm has a competitive ratio of at most $3/2$ in the random-order model, and conjectured the ratio to be $\approx 1.15$. 
However, it has been a long-standing open problem to break the barrier of
$3/2$ even for special cases. Recently, Albers et al.~[Algorithmica '21] showed an improvement 
by proving that in the special case when all the item sizes are greater than $1/3$,
\bestfit{} has a competitive ratio of at most $5/4$ in the random-order model.
In this work, we settle this special case by
showing that Best-Fit has a competitive ratio of exactly $1$, i.e., \bestfit{} performs almost optimally in this special case in the random-order model.
We also make further progress by breaking the barrier of $3/2$ for
the {\em 3-Partition} problem, a notoriously hard special case of bin packing, where all item sizes lie in $(1/4,1/2]$.
\end{abstract}
\maketitle
\setcounter{page}{1}
\section{Introduction}
Bin Packing (BP) is a fundamental NP-hard combinatorial optimization problem. In BP, we are given a set 
$I$ of $n$ items where the $i\Th$ item has weight (also called size) $x_i\in(0,1]$
and the goal is to partition $I$ into the minimum number of sets (bins) such that the total weight of each set is at most 1. 
The problem has numerous applications in logistics, scheduling, cutting stock, etc. \cite{coffman2013bin}. Theoretically, bin packing has been the cornerstone for approximation and online algorithms and the study of the problem has led to the development of several interesting techniques \cite{KarmarkarK82,VegaL81,lee-lee}.

Generally, the performance guarantee of an offline (resp. online) bin packing algorithm $\mathcal{A}$ is measured by asymptotic approximation ratio (AAR) (resp. competitive ratio (CR)).
Let $\OPT(I)$ and $\mathcal{A}(I)$ be the objective values returned by the optimal (offline) algorithm and algorithm $\mathcal{A}$, respectively, on an input $I$.
Then AAR (resp. CR) is defined as
\[
    R_\mathcal{A}^{\infty}\coloneqq \limsup\limits_{m \rightarrow \infty}\left(\sup_{I:\OPT(I)=m} \frac{\mathcal{A}(I)}{\OPT(I)}\right).
\]
Note that $R_\mathcal{A}^{\infty}$ focuses on instances where $\OPT(I)$ is large and avoids pathological instances with large approximation ratios where $\OPT(I)$ is small.

Best-Fit (BF), First-Fit (FF), and Next-Fit (NF) are the three most commonly used algorithms for BP. 
Given $x_i$ as the current item to be packed, they work as follows:
\begin{itemize}
\item BF: Pack $x_i$ into the {\em fullest possible} bin; open a new bin if necessary.
\item FF: Pack $x_i$ into the {\em first possible} bin; open a new bin if necessary.
\item NF: Pack $x_i$ into the {\em most recently opened} bin; open a new bin if necessary.
\end{itemize}

Johnson et al.~\cite{johnson1974worst} studied several heuristics for bin packing such as 
Best-Fit (BF), First-Fit (FF), Best-Fit-Decreasing (BFD), First-Fit-Decreasing (FFD) and showed their (asymptotic) approximation guarantees to be $17/10, 17/10, 11/9, 11/9$, respectively.
Bekesi et al.~\cite{bekesi20005} gave an $O(n)$ time 5/4-asymptotic approximation algorithm. Another $O(n \log n)$ time algorithm is Modified-First-Fit-Decreasing (MFFD) \cite{mffd} which attains an AAR of $71/60\approx 1.1834$.
Vega and Lueker \cite{VegaL81} gave an asymptotic fully  polynomial-time approximation scheme (AFPTAS) for BP: For any $\eps\in(0,1/2)$, it returns a solution with at most $(1+\eps)\OPT(I)+O(1)$
\footnote{In bin packing and related problems, the accuracy parameter $\eps$ is assumed to be a constant. Here, the term $O(1)$ hides some constants depending on $\eps$.}
bins in time $C_\eps+C n \log{1/\eps}$, where $C$ is an absolute constant and $C_\eps$ depends only on  $\eps$.
Shortly after that, Karmarkar and Karp \cite{KarmarkarK82} gave an algorithm that returns a solution using $\OPT(I)+O(\log^2 \OPT(I))$ bins.
The present best approximation is due to Hoberg and Rothvoss \cite{DBLP:conf/soda/HobergR17} which returns a solution using $\OPT(I)+O(\log \OPT(I))$ bins.

The 3-Partition problem
is a notoriously hard special case of bin packing where all item sizes are larger than  $1/4$.
Eisenbrand et al.  \cite{EisenbrandPR11}  mentioned that ``much of the hardness of bin packing seems to appear already
in the special case of 3-Partition when all item sizes are in $(1/4,1/2]$''. 
This problem has deep connections with
Beck's conjecture in discrepancy theory  \cite{spencer1994ten, NewmanNN12}. 
In fact, Rothvoss \cite{DBLP:conf/soda/HobergR17}
conjectured that these  3-Partition instances are indeed the hardest
instances for bin packing and the additive integrality gap of the bin packing configuration LP for these 3-Partition instances is already  $\Theta(\log n)$.

In online BP, items appear one by one and are required to be packed immediately and irrevocably.
Lee and Lee \cite{lee-lee} presented the Harmonic algorithm with competitive ratio $T_{\infty} \approx 1.691$, which is optimal for $O(1)$ space algorithms.
For general online BP, the present best upper and lower bounds for the CR are 1.57829 \cite{BaloghBDEL18} and 1.54278 \cite{BaloghBDEL19}, respectively.

In this paper, we focus on online BP under a stochastic setting called the {\em i.i.d.~model} \cite{coffman1993probabilistic} where the input items are sampled from a sequence of independent and identically distributed (i.i.d.) random variables.
Here, the performance of an algorithm is measured by the 
expected competitive ratio (ECR) 
\[
ER_\mathcal{A} := \lim_{n\to\infty}\frac{\mathbb{E}[\mathcal{A}(I_n(F))]}{\mathbb{E}[\OPT(I_n(F))]},
\] 
where $I_n(F):=(X_1, X_2, \dots, X_n)$ is a list of $n$ random variables drawn i.i.d.~according to some unknown distribution $F$ with support in $(0,1]$.
Mostly, bin packing has been studied under continuous uniform (denoted by $U[a,b], 0\le a<b \le 1$, where item sizes are chosen uniformly from $[a,b]$) or discrete uniform distributions (denoted by $U\{j,k\}, 1\le j \le k$, where item sizes are chosen uniformly from $\{1/k, 2/k, \dots, j/k\}$).
For $U[0,1]$, Coffman et al.~\cite{coffman1980stochastic} showed that NF has an ECR of 4/3 and
Lee and Lee \cite{LL87} showed that the Harmonic algorithm has an ECR of  $\pi^2/3-2\approx 1.2899$.
Interestingly, Bentley et al.~\cite{bentley1984some} showed that the ECR of FF as well as BF converges to $1$ for $U[0,1]$.
It was later shown that the expected wasted space (i.e., the number of needed bins minus  the total size of items) is $\Theta(n^{2/3})$  for First-Fit \cite{shor1986average, coffman1997bin}  and $\Theta(\sqrt{n}\log^{3/4}n)$ for Best-Fit \cite{shor1986average, leighton1989tight}.
Rhee and Talagrand \cite{rhee1993lineB} exhibited an algorithm that, w.h.p., achieves a packing in  $\OPT+O(\sqrt{n}\log^{3/4}n)$ bins for any  distribution $F$ on $(0,1]$.
However, note that their competitive ratio can be quite bad when $\OPT\ll n$.
A distribution $F$ is said to be {\em perfectly packable} if the expected wasted space in the optimal solution
is $o(n)$ (i.e., nearly all bins in an optimal packing are almost fully packed).
Csirik et al.~\cite{csirik2006sum} studied the Sum-of-Squares (SS) algorithm and showed that for any perfectly packable distribution, the expected wasted space is $O(\sqrt{n})$.
However, for distributions that are not perfectly packable, the SS algorithm has an ECR of at most $3$ and can have an ECR of $3/2$ in the worst-case \cite{csirik2006sum}.
For any discrete distribution, they gave an algorithm with an ECR of $1$ that runs in pseudo-polynomial time in expectation.
Gupta et al.~\cite{Gupta020} also obtained similar $o(n)$ expected wasted space guarantee by using an algorithm inspired by the interior-point (primal-dual) solution of the bin packing LP.
However, it remains an open problem to obtain a polynomial-time $(1+\eps)$-competitive algorithm for online bin packing under the i.i.d.~model for arbitrary general distributions.
In fact, the present best polynomial-time algorithm for bin packing under the i.i.d.~model is BF which has an ECR of at most 3/2.
However, Albers et al.~\cite{albers_et_al_MFCS} showed that BF has an ECR $\ge 1.1$ even for a simple distribution: when each item has size $1/4$ with probability $3/5$ and size $1/3$ with probability $2/5$.

We also study the {\em random-order model}, where the adversary specifies the items, but the arrival order is permuted uniformly at random.
The performance measure in this model is called asymptotic random order ratio (ARR): 
\[
RR_{\mathcal{A}}^{\infty} := \limsup\limits_{m \rightarrow \infty}\left(\sup_{I:\OPT(I)=m} \frac{\mathbb{E}[\mathcal{A}(I_{\sigma})]}{\OPT(I)}\right).
\]
Here, $\sigma$ is drawn uniformly at random from $\mathcal{S}_n$, the set of permutations of $n$ elements, and $I_{\sigma}:=(x_{\sigma(1)}, \dots, x_{\sigma(n)})$ is the permuted list.
The random-order model generalizes  the i.i.d.~model \cite{albers_et_al_MFCS}, thus the lower bounds in the random-order model can be obtained from the i.i.d.~model.
Kenyon in her seminal paper \cite{DBLP:conf/soda/Kenyon96}
studied Best-Fit under random-order 
and showed that $1.08 \le RR_{BF}^{\infty} \le 3/2$. Kenyon
conjectured ``$RR_{BF}^{\infty}$ lies somewhere around $1.15$''. The conjecture, if true, raises the
possibility of a better alternative practical offline algorithm: first shuffle the items randomly, then apply Best-Fit. This then beats the AAR of $71/60\approx 1.18$ of the present best practical algorithm MFFD.
The conjecture has received a lot of attention in the past two decades and,
only recently, the analysis on the upper bound has been improved by \cite{soda-gadgets}, who showed that $RR_{BF}^{\infty} < 3/2-10^{-10}$.
They also showed an improved lower bound of $1.144$ using a computer program to solve a large markov chain.
Coffman et al.~\cite{DBLP:journals/dam/CoffmanCRZ08} showed that $RR_{NF}^{\infty} =2$.
Fischer and R{\"{o}}glin \cite{DBLP:conf/latin/FischerR18} achieved analogous results for Worst-Fit \cite{DBLP:journals/jcss/Johnson74} and Smart-Next-Fit \cite{DBLP:journals/ipl/Ramanan89}. Recently, Fischer \cite{fischer_thesis} presented  an exponential-time algorithm, claiming an ARR of $(1+\eps)$.

Several other problems have been studied under the i.i.d.~model and the random-order model
\cite{dean2008approximating,gupta2021stochastic,feldman2009online,gupta2012online,ferguson1989solved,AlbersKL21,DBLP:conf/latin/FischerR16,MahdianY11,Gupta020}.

Monotonicity is a natural property of BP algorithms, which holds if the algorithm never uses fewer
bins to pack $\hat{I}$ when compared $I$, where $\hat{I}$ is obtained from $I$ by increasing the item sizes.
Murgolo \cite{DBLP:journals/dam/Murgolo88} showed that while NF is monotone, BF and FF are not.
\subsection{Our Contributions}
\noindent {\bf Bin packing under the i.i.d.~model:} We achieve a near-optimal performance guarantee for  the bin packing problem under the i.i.d.~model, thus settling the problem. For any arbitrary unknown distribution $F$ on $(0,1]$,  we give a  meta-algorithm (see \cref{sec:iid-model}) that takes an $\alpha$-{asymptotic} approximation algorithm as input and provides a polynomial-time $(\alpha+\eps)$-competitive algorithm. Note that both the distribution $F$ as well as the number of items $n$ are unknown in this case.
We also remark that the distribution $F$ can depend on $n$.

\begin{theorem}
\label{thm:bpiid}
Let $\eps\in(0,1)$ be a constant parameter.
For online bin packing under the i.i.d.~model, where $n$ items are sampled from an unknown distribution $F$,
given an offline algorithm $\mathcal{A}_\alpha$ with an AAR of $\alpha$ and runtime $\beta(n)$,
there exists a meta-algorithm which
returns a solution with an ECR of $(\alpha+\eps)$ and runtime
$O(\beta(n))$.
\footnote{As mentioned in an earlier footnote, here the $O(\cdot)$ notation hides some constants depending on $\eps$.}
\end{theorem}
Using an AFPTAS for bin packing (e.g. \cite{VegaL81}) as $\mathcal A_\alpha$, we obtain the following corollary.
\begin{corollary}
Using an AFPTAS for bin packing as $\mathcal A_\alpha$ in \cref{thm:bpiid}, we obtain an algorithm for online bin packing under the i.i.d.~model
with an ECR of $(1+\eps)$ for any $\eps\in(0,1/2)$.
\end{corollary}

Most algorithms for bin packing under the i.i.d. model are based on the following idea.
Consider a sequence of $2k$ items where each item is independently drawn from an unknown distribution $F$, and let $\mathcal{A}$ be a packing algorithm. Pack the first $k$ items using $\mathcal{A}$; denote the packing by $\mathcal{P}'$. Similarly, let $\mathcal{P}''$ be the packing of the next $k$ items using $\mathcal{A}$. Since each item is drawn independently from $F$, both $\mathcal{P}'$ and $\mathcal{P}''$ have the same properties in expectation;
in particular,
the expected number of bins used in $\mathcal{P}'$ and $\mathcal{P}''$ are the same. 
Thus, intuitively, we want to use the packing   $\mathcal{P}'$ as a proxy for the packing  $\mathcal{P}''$.
However, there are two problems. First, we do not know $n$,
which means that there is no way to know what a good sample size is.
Second, we need to show the stronger statement that w.h.p.~$\mathcal{P}'\approx \mathcal{P}''$.
Note that the items in $\mathcal{P}'$ and $\mathcal{P}''$ are expected to be similar, but they may not be the same.
So, it is not clear which item in $\mathcal{P}'$ is to be used as a proxy for a newly arrived item in the second half.
Due to the online nature, erroneous choice of  proxy items can be quite costly.
Different algorithms handle this problem in different ways.
Some algorithms exploit the properties of particular distributions, some use exponential or pseudo-polynomial time, etc.

Rhee and Talagrand \cite{Rhee_Talagrand_Matching, rhee1993lineB} used {\em upright matching} to decide which item can be considered as a proxy for a newly arrived item.

They consider the model packing $\mathcal P_k$ of the first $k$ items (let's call these the proxy items) using an offline algorithm. With the arrival of each of the next $k$ items, they take a proxy item at random and pack it according to the model packing. Then, they try to fit in the real item using upright matching.
They repeat this process until the last item is packed.
However, they could only show a guarantee of $\OPT+O(\sqrt{n}\log^{3/4}n)$. 
The main drawback of \cite{rhee1993lineB} is that their ECR can be quite bad
if $\OPT\ll n$ (say, $\OPT= n^{2/3}$).
One of the reasons for this drawback is that they don't distinguish between small and large items;
when there are too many small items, the ECR blows up.

Using a similar approach, Fischer \cite{fischer_thesis} obtained a $(1+\eps)$-competitive randomized algorithm for the random-order model,
but it takes exponential time, and the analysis is quite complicated.
The exponential time was crucial in finding the optimal packing which was then used as a good proxy packing.
However, prior to our work, no polynomial-time algorithm existed which achieves a $(1+\eps)$ competitive ratio.

To circumvent these issues, we treat large and small items separately.
However, a straightforward adaptation faces several technical obstacles. 
Thus our analysis required intricate applications of concentration inequalities and sophisticated use of upright matching. 
First, we consider the semi-random case when we know $n$.
Our algorithm works in stages.
For a small constant $\delta \in (0,1]$, the first stage contains only
$\delta^2 n$ items. These items give us an estimate of the distribution.
But since the first stage contains a very small fraction $(\delta^2)$ of the entire input, 
we use a simple algorithm by \nextfit{} to pack it.
If the packing (of the first stage) does not contain too many large items, 
we show that the simple \nextfit{} algorithm suffices for the entire input.
Otherwise, we use a proxy packing of the set of first $\delta^2n$ items to pack the next $\delta^2n$ items.
In the process, the small and large items are packed in a different manner.
{
	The third set of $\delta^2n$ number of items are packed using the proxy packing of the second set of $\delta^2n$
	number of items. This process continues until all the items arrive.
}

{
	Finally, we get rid of the assumption that we know $n$ by first guessing the value of $n$ and then refining our
	guess if it is incorrect. First, we guess the value of $n$ to be a constant $n_0$. If it is incorrect,
	we increase our guess by multiplying $n_0$ with a small factor greater than 1. We continue this process of
	improving our guess until all the items arrive.
}

Our algorithm is simple, polynomial-time (in fact, $O(n)$ time),
and achieves essentially the best possible competitive ratio.
It is relatively simpler to analyze when compared to Fischer's algorithm \cite{fischer_thesis}. 
Also, unlike
the algorithms of Rhee and Talagrand \cite{rhee1993lineB} as well as Fischer \cite{fischer_thesis}, our algorithm is deterministic.
This is because, unlike their algorithms, instead of taking proxy items at random,
we pack all the proxy items before the start of a stage and try to fit in the real items as they come.
This makes our algorithm deterministic.
Our algorithm is explained in detail in \cref{algorithm}.
The nature of the meta-algorithm provides flexibility and ease of application.
See Table \ref{tab:my_label} for the performance guarantees obtained using different offline algorithms. 
\begin{table}[H]
    \centering
    \begin{tabular}{|l|c|c|}
    \hline
          \textbf{$\mathcal{A}_{\alpha}$} & \textbf{Time Complexity} & \textbf{Expected Competitive Ratio} \\ \hline\hline
            AFPTAS \cite{VegaL81} & 
            $O(C_\eps+C n \log{1/\eps})$
            & $(1+\epsilon)$ \\ \hline
            Modified-First-Fit-Decreasing \cite{mffd} & $O(n \log n)$ & $(71/60+\epsilon)$ \\
          \hline
            Best-Fit-Decreasing \cite{johnson-thesis} & $O(n \log n)$ & $({11}/{9}+\epsilon)$ \\ \hline
          First-Fit-Decreasing \cite{johnson-thesis} & $O(n \log n)$ & $({11}/{9}+\epsilon)$ \\ \hline
          Next-Fit-Decreasing \cite{nfd} & $O(n \log n)$ & $(T_\infty+\epsilon)$ \\ \hline
          Harmonic \cite{lee-lee} & $O(n)$ & $(T_\infty+\epsilon)$ \\ \hline
            Next-Fit & $O(n)$ & $(2+\epsilon)$ \\ \hline
    \end{tabular}
    \caption{Analysis of our meta-algorithm depending on $\mathcal{A}_{\alpha}$.
    	In the first row, $C$ is an absolute constant and $C_\eps$ is a constant that depends on $\eps$.
    }
    \label{tab:my_label}
\end{table}
See \cref{sec:iid-model} for the details of the proof and the description of our algorithm.
In fact, our algorithm can easily be generalized to $d$-dimensional online vector packing \cite{BansalE016}, a multidimensional generalization of bin packing. See Section \ref{sec:ovp} for a $d(\alpha+\eps)$ competitive algorithm for $d$-dimensional online vector packing
where the $i\Th$ item $X_i$ can be seen as a tuple $\left(X_i^{(1)},X_i^{(2)},\dots,X_i^{(d)}\right)$ where each $X_i^{(j)}$ is independently sampled from an unknown
distribution $\mathcal D^{(j)}$.

\noindent {\bf Bin packing under the random-order model:}
Next, we study BP under the random-order model.
Recently, Albers et al.~\cite{albers_et_al_MFCS} showed that BF is monotone if all the item sizes are greater than 1/3. Using this result,
they showed that in this special case, BF has an ARR of at most $5/4$. We show that, somewhat surprisingly, in this case, BF actually has an ARR of 1 (see \cref{bfgt13} for the detailed proof).
\begin{theorem}
\label{thm:gt1by3}
For online bin packing under the random-order model, Best-Fit achieves an asymptotic random-order ratio of $1$ when all the item sizes are in $(1/3,1]$.
\end{theorem}

Next, we study
the $3$-partition problem, a special case of bin packing
when all the item sizes are in $(1/4,1/2]$.
This is known to be an extremely hard case \cite{DBLP:conf/soda/HobergR17}.
Albers et al. \cite{albers_et_al_MFCS} mentioned that ``it is sufficient to have one item in $(1/4,1/3]$ to force Best-Fit into anomalous behavior.''
 E.g., BF is non-monotone in the presence of items of size less than $1/3$.
 Thus the techniques of \cite{albers_et_al_MFCS} do not extend to the 3-Partition problem. 
 We break the barrier of $3/2$ in this special case, by showing that BF attains an ARR of at most $1.49107$.

\begin{theorem}
\label{thm:bfroa}
For online bin packing under the random-order model, Best-Fit has an asymptotic random-order ratio of at most $1.49107$ when all the item sizes are in $(1/4,1/2]$.
\end{theorem}
We prove \cref{thm:bfroa} in \cref{sec:roa}.
As 3-partition instances are believed to be the hardest instances for bin packing, our result gives a strong indication that the ARR of BF might be strictly less than $3/2$.
In affirmation, recently Hebbar et al. \cite{soda-gadgets} have managed to show that the ARR of BF is at most $3/2-\eps$ for some $\eps>10^{-10}$.


\section{Notations \& Preliminaries}
Let $I$ denote a list of $n$ items where each item $x$ has an associated weight given by $\weight x$.
We overload this notation and define $\weight I$ as the sum of weights of all the items in $I$.
For any $n\in\mathbb N_+$ we denote the set $\{1,2,\dots,n\}$ by $[n]$. Let $\sigma:[n]\mapsto[n]$ denote a
permutation of $[n]$. Then $I_\sigma$ denotes the list $I$ permuted according
to $\sigma$, i.e., if $x_i$ denotes the $i\Th$ item in $I$ ($i\in[n]$), then
the $i\Th$ item in $I_{\sigma}$ is given by $x_{\sigma(i)}$.



\subsection{The Upright Matching Problem}
The upright matching problem was first introduced in \cite{spaccamela} in the context of designing bin packing algorithms when the items
are sampled from the uniform distribution. In upright matching, a set of $m$ plus ($+$) points and a set of $m$ minus ($-$) points
are located in $\mathbb{R}^2$. Based on the configuration of the plus and minus points, we can create a graph $\mathcal G$
as follows. The vertices of $\mathcal G$ are the plus and minus points. There exists an edge between two vertices $(p^+,p^-)$ iff
\begin{itemize}
    \item The point $p^+$ is a plus point and the point $p^-$ is a minus point.
    \item If $p^+\eqqcolon (x^+,y^+)$ and $p^-\eqqcolon (x^-,y^-)$, then $x^+\ge x^-$ and $y^+\ge y^-$, i.e., the plus point
    $p^+$ lies to the \emph{upright} of the minus point $p^-$.
\end{itemize}
The objective of the upright matching problem is to find the maximum matching in the defined graph $\mathcal G$.
In other words, the objective is to minimize the number of unmatched points. 
Suppose $P^+$ denotes the set of plus points and $P^-$ denotes the set of minus points.
We denote the minimum possible number of unmatched points by $U(P^+,P^-)$.

It can be easily shown that for a given set of plus points $P^+$ and a set of minus points $P^-$, the minimum possible
number of unmatched points can be found in the following way: We process the points in the increasing order of their 
$x$-coordinates. If a plus point $P^+$ arrives, we check if there exists an unmatched minus point that has already arrived and whose
$y$-coordinate is less than that of $P^+$. If no such minus point exists, then $P^+$ is left unmatched for the rest of the process.
Otherwise, we match it with an already arrived, yet, unmatched minus point $P^-$ with the maximum $y$-coordinate among those
whose $y$-coordinates are at most that of $P^+$. A more formal pseudo-code of this procedure is given in \cref{alg:uprightmatching}.

\begin{algorithm}
\caption{Maximum Upright Matching}
\begin{algorithmic}[1]
\Require A set $P^+$ of $m$ plus points and a set $P^-$ of $m$ minus points.
\State Arrange $P^+\cup P^-$ in increasing order of their $x$-coordinates to obtain a sequence $P$.
\State Initialize the set of unmatched minus points $U^-=\emptyset$.
\For{a point $p\eqqcolon (x,y)$ in $P$}
    \If{$p\in P^-$, i.e., $p$ is a minus point}
        \State $U^-\leftarrow U^-\cup\{p\}$
    \ElsIf{$p\in P^+$, i.e., $p$ is a plus point}
        \State $V^-\leftarrow \{(x^-,y^-)\in U^-:y^-\le y)\}$
        \If{$V^-=\emptyset$}
            \State Leave $p$ to be unmatched
        \Else
            \State Define $p^-=\argmax_{(x^-,y^-)\in V^-}y^-$
            \State Match $p^-$ and $p$
            \State $U^-\leftarrow U^-\setminus\{p^-\}$
        \EndIf
    \EndIf
\EndFor
\end{algorithmic}
\label{alg:uprightmatching}
\end{algorithm}
It can be checked that \cref{alg:uprightmatching} correctly computes a maximum upright matching. 
\kvn{A high-level proof can be found in \cite{spaccamela}. However, for the sake of completeness, we provide a detailed proof in \cref{app:upright-matching}.}
\subsubsection{Stochastic Variants of Upright Matching}
We now discuss two stochastic variants of the upright matching problem, which, in one form or the another, have been used in
various works on bin packing; e.g., \cite{spaccamela,shor1986average,rhee1993lineB,fischer_thesis}.
\begin{lemma}
[Upright matching with \iid{} coordinates \cite{fischer_thesis}]
\label{lem:upright-matching-iid}
Consider a set of real numbers $r_1,r_2,\dots,r_m,s_1,s_2,\dots,s_m$ which are 
independently and identically sampled from a distribution. (We remark that the distribution can
be parameterized by $m$.) Define the set of minus points $P^-=\{(-1,r_i)\}_{i\in[m]}$,
and the set of plus points $P^+=\{(+1,s_i)\}_{i\in[m]}$. Then, there exist universal constants
$a,C,K$ such that with probability at least $1-C\exp\left(-a(\log m)^{3/2}\right)$, we have that
\begin{align*}
U(P^+,P^-)\le K\sqrt m(\log m)^{3/4}.
\end{align*}
\end{lemma}
\begin{restatable}[Upright matching with randomly permuted coordinates \cite{fischer_thesis}]{lemma}{uprightmatchingroa}
\label{lem:upright-matching-roa}
Consider a set of reals $r_1,r_2,\dots,r_m$ and $s_1,s_2,\dots,s_m$ 
such that $r_i\le s_i$ for all $i\in[m]$. Consider a uniform random permutation $\pi$ of $[2m]$.
Define the set of minus points $P^-=\{(\pi(i),r_i)\}_{i\in[m]}$,
and the set of plus points $P^+=\{(\pi(m+i),s_i)\}_{i\in[m]}$.
Then, there exist universal constants
$a,C,K$ such that with probability at least $1-C\exp\left(-a(\log m)^{3/2}\right)$, we have that
\begin{align*}
U(P^+,P^-)\le K\sqrt m(\log m)^{3/4}.
\end{align*}
\end{restatable}
\section{Online Bin Packing Problem under the \iid{} Model}
\label{sec:iid-model}
In this section, we provide the meta algorithm as described in \cref{thm:bpiid}.
For the ease of presentation, we split the section into three subsections.
In \cref{sec:blueprint}, we describe `Blueprint Packing', which is one of the central ideas of the paper.
In \cref{algorithm}, we assume a semi-online model, i.e., we assume that the number of items $n$ is known beforehand, and design an algorithm.
The idea of blueprint packing is used extensively as a subroutine in designing this algorithm.
Later, in \cref{sec:imp-algorithm}, we get rid of the assumption on the knowledge of the number of items using a \emph{doubling trick}.

Let the underlying distribution be $F$.
Without loss of generality, we assume that the support set of $F$
is a subset of $(0,1]$.
For any set of items $J$, we define $\weight J$ as the sum of weights of all the items in $J$.
For any $k\in\mathbb N_+$, we denote the set $\{1,2,\dots,k\}$ by $[k]$.
Let $\auxalgo$ be an offline algorithm for bin packing with
an AAR of $\alpha>1$ and let $\OPT$ denote the
optimal algorithm.
Let $\epsilon\in(0,1)$ be a constant parameter and let $0<\delta<\epsilon/8$
be a constant
{such that $1/\delta$ is an integer}.
For any item $x$, we call $x$ to be a \emph{large} item if $x\ge \delta$
and a \emph{small} item otherwise.

\subsection{Blueprint Packing}
\label{sec:blueprint}
One of the central ideas of our algorithm for online bin packing under the \iid{} model is
\emph{blueprint packing}. Informally, it shows us how to pack a set of $k$ \iid{} items arriving online using the 
knowledge of a set of $k$ items that are already present and are sampled from the same distribution independently.

Consider a set $J_1$ of $k$ items ($k\in\mathbb N$) sampled independently from an arbitrary distribution.
Consider an offline bin packing algorithm $\mathcal A_\alpha$ that has an AAR of $\alpha$, and suppose we have
the packing $\mathcal A_\alpha(J_1)$ at our disposal. Now, suppose that another set $J_2$ of $k$ items sampled
independently from the same distribution arrive online. We would like to pack the set $J_2$ (online) using
the packing $\mathcal A_\alpha(J_1)$ as a blueprint.
\begin{lemma}[Blueprint Packing Lemma]
\label{lem:blueprint}
Consider any offline bin packing algorithm $\mathcal A_\alpha$ with an AAR of $\alpha$.
Let $\mathcal F$ be any arbitrary distribution with its support in $(0,1]$ and let $J_1$
be a set of $k$ items sampled from $\mathcal F$ independently.
Consider another set $J_2$ of $k$ items sampled independently from $\mathcal F$ and let $\delta>0$ be a small constant.
Then, there exists an online algorithm that packs $J_2$ in at most
\[
    \alpha(1+4\delta)\expec{\opt(J_2)}+o(\expec{\opt(J_2)})
\]
number of bins, with high probability.
\end{lemma}
We next describe the blueprint packing procedure and prove \cref{lem:blueprint}.
To pack the set $J_2$, we first compute the packing $\mathcal A_\alpha(J_1)$
and refer to it as the \emph{blueprint packing}.
We refer to the large items in the blueprint packing $\mathcal A_\alpha(J_1)$ as \emph{proxy} items.
First, we remove all the small items in the blueprint packing $\mathcal A_\alpha(J_1)$. Then, in each bin of the
modified packing $\mathcal A_\alpha(J_1)$, the empty space is designated to be an \emph{$S$-slot}. The $S$-slots will be used to pack the
small items of $J_2$. (These $S$-slots can be thought of as empty bins of varying sizes to pack the small items of $J_2$.) When a small item $s$ of $J_2$
arrives, we pack it in these $S$-slots using \nextfit{}. If $s$ does not fit according to the \nextfit{} rule,
then we open a new bin and
designate the entire bin as an $S$-slot, and pack $s$ in there. On the other hand, if a large item $\ell$ 
in $J_2$ arrives, we remove the smallest proxy item in the blueprint packing $\mathcal A_\alpha(J_1)$ that is bigger than $\ell$,
if one exists. If no such proxy item exists, we open a new bin for $\ell$, pack it there, and close the bin, i.e., this bin
will not be used for any future items. (As a side note, we can possibly use this bin more efficiently, but we will see that it will not
affect the performance ratio.) A pseudo-code for the blueprint packing procedure can be found in \cref{alg:blueprint}.

\begin{algorithm}
\caption{Blueprint Packing Procedure $\bluep(J_2,\mathcal A_\alpha, J_1)$: Pack the set $J_2$ online using the packing $\mathcal A_\alpha(J_1)$
as a blueprint.}
\begin{algorithmic}[1]
\Require Two disjoint sets of items $J_1$ and $J_2$ where each item in $J_1\cup J_2$ is sampled in an \iid{} manner.
\Ensure Pack the set $J_2$ online using the packing of the set $J_1$ by an offline algorithm $\mathcal A_\alpha$.
\State Construct the blueprint packing $\mathcal A_\alpha(J_1)$.
\State Initialize the set of $S$-slots $\mathcal S= \emptyset$.
\For{each bin $B$ in $\mathcal A_\alpha(J_1)$}
    \State Remove the small items in $B$.
    \State Create an $S$-slot $H$ of size equal to ($1-$weight of all the large items in $B$).
    \State $\mathcal S\leftarrow\mathcal S\cup \{H\}$.
\EndFor
\State Initialize the set of proxy items $D$ to the set of large items in the set $J_1$. \Comment{The set of proxy items.}
\For{an item $x \in J_2$}
    \If{$x$ is large}
        \If{there exists a proxy item $d\in D$ such that $d\ge x$} \label{line:check-proxy-item}
            \State Find smallest such $d$.
            \State $D\leftarrow D\setminus \{d\}$.
            \State Pack $x$ in place of $d$ in the packing $\mathcal A_\alpha(J_1)$.
        \Else
            \State Open a new bin and pack $x$ and close the bin. \label{line:pack-unmatched}
        \EndIf
    \Else{ (i.e., $x$ is small)}
        \State Try packing $x$ in the set of $S$-slots $\mathcal S$ using \nextfit{}.
        \If{$x$ cannot be packed}
            \State Open a new bin $B$ with a single $S$-slot of unit capacity.\label{line:overflowed-small}
            \State $\mathcal S\leftarrow\mathcal S\cup \{B\}$.
            \State Pack $x$ in $B$.
        \EndIf
    \EndIf
\EndFor
\State\Return the packing after removing the proxy items that have not been replaced.
\end{algorithmic}
\label{alg:blueprint}
\end{algorithm}
\subsubsection{Concentration Bounds to Analyze Blueprint Packing}
We will need some concentration inequalities and tail bounds to analyze the blueprint packing procedure.
First, we state the Bernstein's inequality, which will be used heavily throughout. The details and proof can be found in \cite{bernstein-ref}.
\begin{lemma}[Bernstein's Inequality]
\label{bernstein}
Let $X_1,X_2,\dots,X_n$ be independent random variables such that each $X_i\in[0,1]$.
Then, for any $\lambda>0$, the following inequality holds.
\begin{align*}
    \prob{\abs{\sum_{i=1}^nX_i-\sum_{i=1}^n\expec{X_i}}\ge \lambda}\le 2\exp\left(-\frac{\lambda^2}{2\left(\sum_{i=1}^n\expec{X_i}+\lambda/3\right)}\right).
\end{align*}
\end{lemma}
The following lemma is a direct implication of the results of~\cite{rhee1993lineB,rhee-ineq}.
\begin{lemma}\label{rhee-talagrand}
For any $t\in [n]$, let $I(1,t)$ denote the first $t$ items of a \kvn{set $I$ of $n$ items sampled independently from a distribution}. Then there exist constants $K,a>0$ such that,
\begin{align*}
\prob{\opt (I(1,t)) \geq \frac{t}{n}\expec{\opt(I)} + K\sqrt{n}(\log n)^{3/4}} \leq \exp\left(-a(\log n)^{3/2}\right).
\end{align*}
\end{lemma}
\begin{proof}[Proof Sketch]
The following claim is similar to Theorem 2.1 in \cite{rhee1993lineB}.
\begin{claim}
\label{claim:rhee-talagrand}
Let $I$ be a list of $n$ items sampled independently from a distribution and let $I(1,t)$
denote the first $t$ items. Then there exist constants $K_1,a_1>0$ such that
\begin{align*}
\prob{\opt(I(1,t))\ge\frac{t}{n}\opt(I)+K_1\sqrt n(\log n)^{3/4}}\le \exp\left(-a_1(\log n)^{3/2}\right)
\end{align*}
\end{claim}
The proof is almost the same as the proof of Theorem 2.1 in \cite{rhee1993lineB} \kvn{with some small changes.
However, since these changes are not very trivial, we give the full proof in \cref{app:proportionality}.}

The next claim is a direct implication of the main result of \cite{rhee-ineq}.
\begin{claim}
Let $I$ be a list of $n$ items sampled independently from a distribution.
Then there exist constants $K_2,a_2>0$ such that
\begin{align*}
\prob{\expec{\opt(I)}\ge\opt(I)+K_2\sqrt n(\log n)^{3/4}}\le \exp\left(-a_2(\log n)^{3/2}\right)
\end{align*}
\end{claim}
Combining both claims, we obtain that there exist constants $K,a>0$ such that
\begin{align*}
\prob{\opt(I(1,t))\ge\frac{t}{n}\expec{\opt(I)}+K\sqrt n(\log n)^{3/4}}\le \exp\left(-a(\log n)^{3/2}\right)
\end{align*}
\end{proof}
The following lemmas are about how a property of a part of the input (say, the total size) compares to that of the entire input.
\begin{lemma}\label{weight-conc}
For an input set $I$ of $n$ items drawn independently from a distribution, for any arbitrary set $J\subseteq I$ we have,
\begin{align*}
\prob{\abs{\weight{J} - \frac{\abs J}{n} \expec{\weight{I}}} \geq \expec{\weight{I}}^{2/3}} \leq 2\exp\left(-\frac{1}{3}\expec{\weight{I}}^{1/3} \right).
\end{align*}
\end{lemma}
\begin{proof}
    We can assume that $\weight{I}$ goes to infinity
since we know that $\opt(I)\le 2\weight{I}+1$.
Let $K:=\expec{\weight{I}}$ and $\weight{x}$ be the weight of item $x$.
Using Bernstein's inequality (\cref{bernstein}),
\begin{align*}
&\prob{\abs{\sum_{x\in J}\weight{x} - \sum_{x\in J}\expec{\weight{x}}} \geq K^{2/3}}\\
            \leq\:&2\exp\left(-\frac{K^{4/3}}{2\left(\sum_{x\in J}\expec{\weight{x}}+K^{2/3}/3\right)}\right)\\
            \le\:&2\exp\left(-\frac{K^{4/3}}{2\left(K+K^{2/3}/3\right)}\right)\\
            \le\:&2\exp\left(-\frac{K^{4/3}}{2\left(K+K/3\right)}\right)\tag{since $K$ goes to infinity, $K^{2/3}\le K$}\\
            \le\:&2\exp\left(-\frac13K^{1/3}\right).
\end{align*}
Since $\sum_{x\in J}\weight{x}=\weight{J}$ and since $\expec{\weight{J}}=\frac{\abs J}{n}\expec{\weight{I}}$,
the lemma follows.
\end{proof}
\begin{lemma}\label{large-item-conc}
Let $I$ be an input set of $n$ items drawn independently from a distribution and let $J$ be any  subset of $I$.
Suppose $J_{\ell}$ (resp. $I_\ell$) denote the set of large items in $J$ (resp. $I$). Then we have,
\begin{align*}
\prob{\abs{\abs{J_{\ell}} - \frac{\abs J}{n}\expec{\abs\largeitems}} \geq \expec{\weight{I}}^{2/3}} \leq 2\exp\left(-\frac{\param}{3}\expec{\weight{I}}^{1/3} \right).
\end{align*}
\end{lemma}
\begin{proof}
    We can assume that $\weight{I}$ goes to infinity.
Let $K:=\expec{\weight{I}}$.
For any item $x$, let $L_x$ be the indicator random variable which denotes if the item $x$
is a large item or not.
Using Bernstein's inequality (\cref{bernstein}),
\begin{align*}
&\prob{\abs{\sum_{x\in J}L_x - \sum_{x\in J}\expec{L_x}} \geq K^{2/3}}\\
            \leq\:&2\exp\left(-\frac{K^{4/3}}{2\left(\sum_{x\in J}\expec{L_x}+K^{2/3}/3\right)}\right)\\
            =\:&2\exp\left(-\frac{K^{4/3}}{2\left(\expec{\abs{J_\ell}}+K^{2/3}/3\right)}\right)
                    \tag{since $\sum_{x\in J}L_x=\abs{J_\ell}$}\\
            \le\:&2\exp\left(-\frac{K^{4/3}}{2\left(K/\delta+K^{2/3}/3\right)}\right)
                    \tag{since $K=\expec{\weight{I}}\ge\delta\expec{\abs{J_\ell}}$}\\
            \le\:&2\exp\left(-\frac\delta3K^{1/3}\right).
\end{align*}
Since $\sum_{x\in J}L_x=\abs{J_\ell}$ and since $\expec{\abs{J_\ell}}=\frac{\abs J}{n}\expec{\abs{I_\ell}}$,
the lemma follows.
\end{proof}
\kvn{The next lemma shows that the optimal number of bins required to pack a set of items sampled independently from a distribution
is approximately proportional to the number of items.}
\begin{lemma}\label{lem:decomposition}
For any $t\in \{1,2,\ldots, n\}$, let $I(1,t)$ denote the first $t$ items of a set $I$ of items sampled independently from a distribution. Then there exist constants $C,a>0$ such that with probability at least  $1-\exp\left(-a\left(\log\expec{\opt(I)}\right)^{1/3}\right)$, we have
\begin{align*}
\opt (I(1,t)) \leq (1+2\param)\frac{t}{n}\expec{\opt(I)} + C\expec{\opt(I)}^{2/3}.
\end{align*}
\end{lemma}
\begin{proof}
    Let us denote the set of large items in $I(1,t)$ by $I_\ell(1,t)$
and denote the set of small items by $I_s(1,t)$. Consider any optimal
packing of $I_\ell(1,t)$. Start packing $I_s(1,t)$ in the spaces left in the
bins greedily using \nextfit{} while opening new bins whenever necessary.
This gives us a valid packing of $I(1,t)$.

    We distinguish between the cases when we open new bins to pack $I_s(1,t)$ and when we do not.
    We will also use \cref{rhee-talagrand,weight-conc,large-item-conc} to bound the number of opened bins.
    Using a union bound, we can assume that there exists a constant $a'$ such that the guarantees given by \cref{rhee-talagrand,large-item-conc,weight-conc}
    hold with probability at least $1-\exp\left(-a'\expec{\opt(I)}^{2/3}\right)$.
    Define the event
    \begin{align*}
        E\coloneqq \text{``We do not open any new bins for $I_s(1,t)$''.}
    \end{align*}
    Suppose the event $E$ occurs with probability at least $\exp(-\frac12 a'(\log \expec{\opt(I)})^{3/2})$.
    Note that since \cref{rhee-talagrand} (unconditionally) holds with probability at least $1-\exp(-\frac12 a'(\log \expec{\opt(I)})^{3/2})$,
    when we condition on the event $E$, it holds with probability at least $1-\frac{\exp(-a'(\log \expec{\opt(I)})^{3/2})}{\exp(-\frac12 a'(\log \expec{\opt(I)})^{3/2})}$
    which is equal to $1-\exp(-\frac12 a'(\log \expec{\opt(I)})^{3/2})$.

Now, we condition on event $E$. In this case, we have
\begin{align}
    \opt(I(1,t))&= \opt(I_\ell(1,t)).\label{eq:I1t-Il1t}
\end{align}
We further consider two cases. 

\textbf{Case 1:} $\expec{\abs{I_\ell}}\le2\expec{\weight{I}}^{2/3}$.\\
Then we have
\begin{align*}
    \opt(I(1,t))&= \opt(I_\ell(1,t))\\
            &\le \abs{I_\ell}\\
            &\le \expec{\abs{I_\ell}}+\expec{\weight{I}}^{2/3}\tag{using \cref{large-item-conc}}\\
            &\le 3\expec{\weight{I}}^{2/3}\\
            &\le 3\expec{\opt(I)}^{2/3}.
\end{align*}    
Now, we consider the other case.

\textbf{Case 2:} $\expec{\abs{I_\ell}}>2\expec{\weight{I}}^{2/3}$.\\
Now, to bound $\opt(I_\ell(1,t))$, we can apply \cref{rhee-talagrand} confined to only large items, by treating other items as having zero weight.
More formally, consider $I$ and remove all the small items in $I$ to obtain $I_\ell$. Then, an item of $I_\ell$ can be thought of as
sampling from a different distribution $F_\ell$ as follows. We sample an item from $F$; if it is a small item, we discard it and sample until we get a large item.

Define
\begin{align*}
    t' = \frac tn\expec{\abs{I_\ell}}+\expec{\weight{I}}^{2/3}\quad\text{and}\quad n'=\expec{\abs{I_\ell}}-\expec{\weight{I}}^{2/3}.
\end{align*}
Owing to \cref{large-item-conc}, we obtain that, with high probability,
\begin{align}
    \abs{I_\ell(1,t)}&\le t',\label{eq:t'-bound}\\
    \abs{I_\ell}&\ge n'\label{eq:n'-bound}.
\end{align}
Define $I'_\ell$ as the prefix of $I_\ell$ with $n'$ items.
Note that $t'$ can be larger than $n'$. (This happens when, e.g., $t=n$.)
Therefore, we consider two further subcases.

\textit{Case 2.1.} $t'\le n'$.\\
In this case, define
\begin{align*}
    {t''\coloneqq \frac tn n'=\frac{t}{n}\expec{\abs{I_{\ell}}}-\frac tn\expec{\weight{I}}^{2/3}},
\end{align*}
and note that
\begin{align}
    t'\le t''+2\expec{\weight{I}}^{2/3}.\label{eq:t''-bound}
\end{align}
Define $I'_\ell(t')$ as the prefix of $t'$ items of $I'_\ell$
and $I'_\ell(t'')$ as the prefix of $t''$ items of $I'_\ell$.
We have that 
\begin{align*}
    \opt(I_\ell(1,t))&\le \opt(I'_\ell(t'))\tag{from \cref{eq:t'-bound}}\\
                    &\le \opt(I'_\ell(t''))+(t'-t'')\\
    &\le \opt(I'_\ell(t''))+2\expec{\weight{I}}^{2/3}.\tag{from \cref{eq:t''-bound}}
\end{align*}
Further, from \cref{eq:n'-bound} we also have
\begin{align*}
    \quad\opt(I'_\ell)\le \opt(I_\ell).
\end{align*}
Therefore,
\begin{align}
\opt(I_\ell(1,t))&\le \opt(I'_\ell(t''))+2\expec{\weight{I}}^{2/3}\nonumber\\
        &\le\frac{t''}{n'}\expec{\opt(I'_\ell)}+C_1\sqrt{n'}(\log n')^{3/4}+2\expec{\weight{I}}^{2/3}
                                                                                            \tag{using \cref{rhee-talagrand}}\nonumber\\
        &=\frac tn\expec{\opt(I'_\ell)}+C_1\sqrt{n'}(\log n')^{3/4}+2\expec{\weight{I}}^{2/3}\nonumber\\
        &\le\frac tn\expec{\opt(I_\ell)}+C_1(n')^{2/3}+2\expec{\weight{I}}^{2/3}\nonumber\\
        &\le\frac tn\expec{\opt(I_\ell)}+C_1\expec{\abs{I_\ell}}^{2/3}+2\expec{\weight{I}}^{2/3}\nonumber.
\end{align}
We now use the facts that $\opt(I_\ell)\le \opt(I)$ and $\delta \abs{I_\ell}\le \weight{I}$ to obtain
\begin{align}
\opt(I(1,t))=\opt(I_{\ell}(1,t))&\le \frac{t}{n}\expec{\opt(I)}+C_1'\expec{\opt(I)}^{2/3}
\end{align}
with probability at least $1-\exp\left(-a''(\log\expec{\opt(I)})^{1/3}\right)$ for some constants $a'',C_1'>0$.
This ends Case 2.1.

\textit{Case 2.2.} $t'> n'$.\\
This is the easier case. Since we have $t'>n'$, by the definitions of $t',n'$, we obtain that
\begin{align*}
    \frac tn\expec{\abs{I_\ell}}+\expec{\weight{I}}^{2/3}>\expec{\abs{I_\ell}}-\expec{\weight{I}}^{2/3}.
\end{align*}
This gives us the inequality
\begin{align}
    \left(1-\frac tn\right)\expec{\abs{I_{\ell}}} < 2\expec{\weight{I}}^{2/3}.\label{eq:case22-main-bound}
\end{align}
Further observe that
\begin{align*}
    \abs{I_{\ell}(1,t)}&\le t'\tag{by \cref{eq:t'-bound}}\\
                    &\le n'+ 2\expec{\weight{I}}^{2/3}\tag{by definition of $n'$}\\
                    &= \abs{I_{\ell}'}+2\expec{\weight{I}}^{2/3}\tag{from \cref{eq:n'-bound}}.
\end{align*}
This implies that
\kvn{
\begin{align*}
    \opt(I_{\ell}(1,t))&\le \opt(I_{\ell}')+2\expec{\weight{I}}^{2/3}\\
                    &\le \expec{\opt(I_{\ell}')}+C_1\sqrt{n'}(\log n')^{3/4}+2\expec{\weight{I}}^{2/3}\tag{using \cref{rhee-talagrand}}\\
                    &\le \expec{\opt(I_{\ell}')}+C_1{(n')}^{2/3}+2\expec{\weight{I}}^{2/3}\\
                    &\le \expec{\opt(I_{\ell}')}+C_1(\expec{\abs{I_{\ell}}})^{2/3}+2\expec{\weight{I}}^{2/3}\tag{since $n'\le \expec{\abs{I_{\ell}}}$}\\
                    &\le \expec{\opt(I_{\ell})}+C_1(\expec{\abs{I_{\ell}}})^{2/3}+2\expec{\weight{I}}^{2/3}\\
                    &= \frac tn \expec{\opt(I_{\ell})}+\left(1-\frac tn\right)\expec{\opt(I_{\ell})}+C_1(\expec{\abs{I_{\ell}}})^{2/3}+2\expec{\weight{I}}^{2/3}\\
                    &\le \frac tn \expec{\opt(I_{\ell})}+\left(1-\frac tn\right)\expec{\abs{I_{\ell}}}+C_1(\expec{\abs{I_{\ell}}})^{2/3}+2\expec{\weight{I}}^{2/3}\\
                    &\le\frac tn \expec{\opt(I_{\ell})}+2\expec{\weight{I}}^{2/3}+C_1\expec{\abs{I_\ell}}^{2/3}+2\expec{\weight{I}}^{2/3}\tag{using \cref{eq:case22-main-bound}}.
\end{align*}
We again use the facts that $\opt(I_\ell)\le \opt(I)$ and $\delta \abs{I_\ell}\le \weight{I}$ to obtain
\begin{align}
\opt(I(1,t))=\opt(I_{\ell}(1,t))&\le \frac{t}{n}\expec{\opt(I)}+C_1'\expec{\opt(I)}^{2/3}\label{zeta-eq-1}
\end{align}
with probability at least $1-\exp\left(-a''(\log\expec{\opt(I)})^{1/3}\right)$ for some constants $a'',C_1'>0$.
This ends Case 2.2, and hence the
analysis in the scenario when the event $E$ happens with probability at least $\exp(-\frac12 a'(\log \expec{\opt(I)})^{3/2})$.
}

On the other hand, suppose the event $E$ happens with probability at most $\exp(-\frac12 a'(\log \expec{\opt(I)})^{3/2})$.
In other words, with probability at least $1-\exp(-\frac12 a'(\log \expec{\opt(I)})^{3/2})$, we open new bins while packing $I_s(1,t)$. Then after the final packing, every bin (except possibly one)
is filled to a level of at least $1-\delta$. Hence, in this case,
\begin{align}
\opt(I(1,t))&\le\frac{1}{1-\param}\weight{I(1,t)}+1\nonumber\\
            &\le(1+2\delta)\frac{t}{n}\expec{\weight{I}}+(1+2\delta)\expec{\weight{I}}^{2/3}+1\tag{using \cref{weight-conc}}\nonumber\\
            &\le(1+2\delta)\frac{t}{n}\expec{\opt(I)}+C_2\expec{\opt(I)}^{2/3}\label{zeta-eq-2}
\end{align}
with probability at least $1-\exp\left(-a'''\expec{\opt(I)}^{1/3}\right)$ for some constants $a''',C_2>0$.

Therefore, if the event $E$ occurs with probability at least $\exp\left(-\frac12 a'\expec{\opt(I)}^{1/3}\right)$, then \cref{zeta-eq-1} applies.
Otherwise, \cref{zeta-eq-2} applies. Hence, we obtain that there exist some constants $a,C$ such that
\begin{align*}
    \opt(I(1,t))\le(1+2\delta)\frac{t}{n}\expec{\opt(I)}+C\expec{\opt(I)}^{2/3}
\end{align*}
holds with probability at least $1-\exp\left(-a\expec{\opt(I)}^{1/3}\right)$.

This completes the proof.
\end{proof}
\subsubsection{Analysis of Blueprint Packing}
Using the above lemmas, we will now analyze \cref{alg:blueprint} to bound the number of bins used to pack $J_2$.
Let us denote the number of bins used by \cref{alg:blueprint} with 
$\bluep(J_2,\mathcal A_\alpha, J_1)$. We interpret this notation as the number of bins used to pack $J_2$ using
the blueprint packing $\mathcal A_\alpha(J_1)$.
Let $\bluep_{\unmatch}$ denote the number of bins opened for the large items in $J_2$ for which
we could not find a proxy item to be replaced. (See \cref{line:pack-unmatched} in \cref{alg:blueprint}.)
Let $\bluep_{\mini}$ denote the number of bins opened because a small item could not be packed in the
$S$-slots using \nextfit{}. (See \cref{line:overflowed-small} in \cref{alg:blueprint}.)
Even before the arrival of the set $J_2$, we have used $\mathcal A_\alpha(J_1)$ number of
bins to construct the blueprint packing. Hence,
\begin{align}
\bluep(J_2,\mathcal A_\alpha, J_1)=\mathcal A_\alpha(J_1)+\bluep_{\unmatch}+\bluep_{\mini}.\label{eq:to-be-minimized}
\end{align}
We will bound each of the summands in the right hand side of the above equation.
First, since $\mathcal A_\alpha$ is a bin packing algorithm with AAR $\alpha$, we have,
\begin{align}
\mathcal A_\alpha(J_1)&\le \alpha \Opt(J_1)+o(\Opt(J_1)).
\label{eq:bluep-existing}
\end{align} 
To bound $\bluep_{\unmatch}$, we use the upright matching result of \cref{lem:upright-matching-iid}.
Let $L(J_1)$ and $L(J_2)$ denote the set of large items in the sequences $J_1$ and $J_2$, respectively,
and let $\kappa_1\coloneqq \abs{L(J_1)}$ and $\kappa_2\coloneqq \abs{L(J_2)}$.
Using \cref{large-item-conc}, we get that, w.h.p., $\kappa_1\geq \expec{\kappa_1} - \expec{\weight{J_1}}^{2/3}$, and
$\kappa_2\leq \expec{\kappa_2} + \expec{\weight{J_2}}^{2/3}$. Since $\abs{J_1}=\abs{J_2}$ and since each item in
$J_1\cup J_2$ is sampled independently and identically, we have that $\expec{\kappa_1}=\expec{\kappa_2}$
and $\expec{\weight{J_1}}=\expec{\weight{J_2}}$. Hence, with high probability,
\begin{align}
\kappa_2 \leq  \kappa_1 + 2\expec{\weight{J_1}}^{2/3}.
\label{eq:largeitems-j}
\end{align}
Let $L(J_1)\eqqcolon \{r_1,r_2,\dots,r_{\kappa_1}\}$ and $L(J_2)\eqqcolon \{s_1,s_2,\dots,s_{\kappa_2}\}$.
According to lines 10--17 of \cref{alg:blueprint}, each $s_i$ ($i\in[\kappa_1]$) is packed in the place of an
$r_j$ $(j\in[\kappa_1])$ satisfying $r_j>s_i$ such that $r_j$ is minimum. Now, if we consider the sets of points
$P^-=\{(-1,r_j)\}_{j\in[\kappa_1]}$, and $P^+=\{(+1,s_i)\}_{i\in[\kappa_1]}$,
then one can see that this procedure of packing $s_i$-s in place of $r_j$-s is exactly the same as the maximum upright matching procedure
detailed in \cref{alg:uprightmatching}. Hence, by \cref{lem:upright-matching-iid} and also
accounting for the items in $\{s_{\kappa_1+1},s_{\kappa_1+2},\dots,s_{\kappa_2}\}$, we obtain,
w.h.p., an upper bound on $\bluep_{\unmatch}$, the number of new bins opened while packing $L(J_2)$.
\begin{align*}
\bluep_{\unmatch}\le \kappa_2-\kappa_1+K\sqrt{\kappa_1}(\log \kappa_1)^{3/4}.
\end{align*}
Using \cref{eq:largeitems-j}, we finally obtain the following bound on $\bluep_{\unmatch}$ that holds with
high probability.
\begin{align}
\bluep_{\unmatch}&\le 2\expec{\weight{J_1}}^{2/3}+K\sqrt{\kappa_1}(\log \kappa_1)^{3/4}.\nonumber
\end{align}
We will simplify the above inequality further.
Each large item has size at least $\delta$. Hence, at most $1/\delta$ number of large items can be fit in a bin.
Therefore, $\Opt(J_1)\ge \delta\kappa_1$. Also, $\Opt(J_1)\ge\weight{J_1}$. Since $\delta$ is a constant,
and since $\log a\le a^{2/9}$ for any positive integer $a$, we have that for some constant $K_1$, \whp{}, 
\begin{align}
\bluep_{\unmatch}\le 2\expec{\opt{J_1}}^{2/3}+K_1\opt(J_1)^{2/3}\nonumber
\end{align}
Using \cref{lem:decomposition}, we see that, \whp{}, $\Opt(J_1)\le (1+2\delta)\expec{\Opt(I)}+o(\expec{\opt(I)})$.
Hence, the above inequality can be transformed as
\begin{align}
\bluep_{\unmatch}=o\left(\expec{\opt(J_1)}\right).
\label{eq:bluep-unmatch}
\end{align}
The only part that is left to be bounded is $\bluep_\mini$, the number of bins opened because the small items in $J_2$
could not be packed in the $S$-slots (lines 20--24 of \cref{alg:blueprint}).
We will bound this by using the concentration of weights of small items in $J_1$ and $J_2$.
For this purpose, let us write $S(J_1), S(J_2)$ to denote the sets of small items in the sets $J_1,J_2$,
respectively.
Also, let $\mathcal S_e$ denote the set of existing $S$-slots before we start packing the items in $J_2$,
i.e., the set of $S$-slots created just after the for loop in lines 3--7 of \cref{alg:blueprint} ends.
Similarly, let $\mathcal S_n$ denote the set of $S$-slots created during the process of packing $J_2$,
i.e., when line 21 of \cref{alg:blueprint} is executed. We would like to bound $\abs{\mathcal S_n}$.
\kvn{To this end, let us define the volume of a given set $\mathcal S$ of $S$-slots
as the sum of the size\footnote{The size of an $S$-slot is the maximum size of an item that can fit in that slot.} of each $S$-slot in $\mathcal S$.
Let $\vol{\mathcal S}$ denote the volume of a given set $\mathcal S$ of $S$-slots.}
Since we packed the small items in $J_2$ in the set of $S$-slots given by $\mathcal S_e\cup \mathcal S_n$ using \nextfit{}, each of these
$S$-slots, except possibly one, will have an unused volume of at most $\delta$. Hence,
\begin{align}
\weight{S(J_2)}&\ge \vol{\mathcal S_n\cup\mathcal S_e}-\delta\abs{\mathcal S_n\cup\mathcal S_e}-1\nonumber\\
            &=   \vol{\mathcal S_n}+\vol{\mathcal S_e}-\delta\abs{\mathcal S_n}-\delta\abs{\mathcal S_e}-1\nonumber\\
            &=   \abs{\mathcal S_n}+\vol{\mathcal S_e}-\delta\abs{\mathcal S_n}-\delta\abs{\mathcal S_e}-1
                            \tag{since each $S$-slot in $\mathcal S_n$ has unit volume}\nonumber\\
            &=   \vol{\mathcal S_e}+(1-\delta)\abs{\mathcal S_n}-\delta\abs{\mathcal S_e}-1.
            \label{dkgjkljgklgklg}
\end{align}
Let us look at the right hand side of the last inequality closely.
The number of $S$-slots in $\mathcal S_e$ is upper bounded by the number of bins used by the algorithm $\mathcal A_\alpha$
to pack $J_1$. At the same time, $\vol{\mathcal S_e}$ is lower bounded by the volume of the small items in $J_1$.
Hence, \cref{dkgjkljgklgklg} can be rewritten as
\begin{align*}
\weight{S(J_2)}&\ge\vol{\mathcal S_e}+(1-\delta)\abs{\mathcal S_n}-\delta\abs{\mathcal S_e}-1\\
               &\ge \weight{S(J_1)}+(1-\delta)\abs{\mathcal S_n}-\delta\mathcal A_\alpha(J_1)-1\\
               &\ge \weight{S(J_1)}+(1-\delta)\abs{\mathcal S_n}-\delta\alpha\Opt(J_1)-o(\Opt(J_1)).
\end{align*}
Rearranging terms on both the sides, we obtain an upper bound on $\abs{\mathcal S_n}$.
\begin{align}
\bluep_{\mini}=\abs{\mathcal S_n}&\le\frac{1}{1-\delta}\Big(\weight{S(J_2)}-\weight{S(J_1)}\Big)+\delta\alpha\Opt(J_1)+o(\Opt(J_1))
\label{eq:bluep-small}
\end{align}
The only task left is to bound the quantity $\weight{S(J_2)}-\weight{S(J_1)}$. Intuitively, this quantity should not
be too big since the joint distributions of $J_2$ and $J_1$ are exactly the same. This is formalized in the next claim.

\begin{claim}
With probability at least $1-4\exp\left(-\frac38\expec{\opt(J_1)}^{1/3}\right)$, we have
\begin{align*} 
\weight{S(J_2)}-\weight{S(J_1)}\le 2\expec{\opt(J_1)}^{2/3}.
\end{align*}
\end{claim}

\begin{proof}
For any item $x$, define the random variable $S_x$ which takes value $0$ if $x$ is large,
and takes value $x$ if $x$ is small. Now, observe that 
$\sum_{x\in J_1}S_x=\weight{S(J_1)}$ and $\sum_{x\in J_2}S_x=\weight{S(J_2)}$. Moreover, since $J_1$ and $J_2$ each
contain $k$ items and since each item is sampled independently from the same distribution, it must be the case that
$\expec{\weight{S(J_1)}}=\expec{\weight{S(J_2)}}$. The same argument implies that
$\expec{\Opt(J_1)}=\expec{\Opt(J_2)}$.

Applying Bernstein's inequality (\cref{bernstein}) on the set of random variables $\{S_x\}_{x\in J_2}$, we obtain that
\begin{align*}
\prob{\weight{S(J_2)}-\expec{\weight{S(J_2)}}\ge \expec{\Opt(J_2)}^{2/3}}&=\prob{\sum_{x\in J_2}S_x-\sum_{x\in J_2}\expec{S_x}\ge \expec{\Opt(J_2)}^{2/3}}\\
                                        &\le 2\exp\left(-\frac{\expec{\Opt(J_2)}^{4/3}}{2\left(\expec{\weight{S(J_2)}}+\frac13\expec{\Opt(J_2)}^{2/3}\right)}\right).
\end{align*}
Now, observe that $\weight{S(J_2)}\le \weight{J_2}\le\Opt(J_2)$. Also, $\expec{\Opt(J_2)}^{2/3}\le \expec{\Opt(J_2)}$. 
Hence, the above inequality can be simplified as
\begin{align*}
\prob{\weight{S(J_2)}-\expec{\weight{S(J_2)}}\ge \expec{\Opt(J_2)}^{2/3}}&\le 2\exp\left(-\frac38 \expec{\Opt(J_2)}^{1/3}\right).
\end{align*}
In the same manner, we can apply Bernstein's inequality on the set of random variables $\{S_x\}_{x\in J_1}$ to obtain that
\begin{align*}
\prob{\weight{S(J_1)}-\expec{\weight{S(J_1)}}\le -\expec{\Opt(J_1)}^{2/3}}&\le 2\exp\left(-\frac38 \expec{\Opt(J_1)}^{1/3}\right).
\end{align*}
Therefore, with probability at least $1-4\exp\left(-\frac38 \expec{\Opt(J_1)}^{1/3}\right)$, both the following inequalities hold.
\begin{align*}
\weight{S(J_2)}-\expec{\weight{S(J_2)}}&\le \expec{\Opt(J_2)}^{2/3},\\
\weight{S(J_1)}-\expec{\weight{S(J_1)}}&\ge -\expec{\Opt(J_1)}^{2/3}.
\end{align*}
Since $\expec{\weight{S(J_1)}}=\expec{\weight{S(J_2)}}$ and $\expec{\Opt(J_1)}=\expec{\Opt(J_2)}$, we obtain that
with probability at least $1-4\exp\left(-\frac38 \expec{\Opt(J_1)}^{1/3}\right)$,
\begin{align*}
\weight{S(J_2)}-\weight{S(J_1)}\le 2\expec{\opt(J_1)}^{2/3}.
\end{align*}
This ends the proof of the claim.
\end{proof}
Using the above claim with \cref{eq:bluep-small}, we can obtain an upper bound on $\bluep_{\mini}$
that holds with high probability as follows.
\begin{align}
\bluep_{\mini}&\le\frac{1}{1-\delta}\left(2\expec{\opt(J_1)}^{2/3}\right)+\delta\alpha\Opt(J_1)+o(\Opt(J_1))\nonumber\\
        &\le o\left(\expec{\opt(J_1)}\right)+\delta\alpha\Opt(J_1)+o(\Opt(J_1)).\nonumber
\end{align}
\cref{lem:decomposition} with $t=n$ tells us that $\Opt(J_1)\le(1+2\delta)\expec{\Opt(J_1)}+o(\expec{\opt(J_1)})$.
Hence, we transform the above inequality to obtain the following upper bound on $\bluep_{\mini}$.
\begin{align}
\bluep_{\mini}\le \delta\alpha\Opt(J_1)+o\left(\expec{\opt(J_1)}\right).
\label{eq:bluep-small-final}
\end{align}
Summing up \cref{eq:bluep-existing,eq:bluep-unmatch,eq:bluep-small-final},
we obtain an upper bound on $\bluep(J_2,\mathcal A_\alpha, J_1)$ (\cref{eq:to-be-minimized}),
the number of bins used to pack $J_2$ using \cref{alg:blueprint}.
\begin{align*}
\bluep(J_2,\mathcal A_\alpha, J_1)&=\mathcal A_\alpha(J_1)+\bluep_{\unmatch}+\bluep_{\mini}\\
                &\le \alpha\opt(J_1)+\alpha\delta\opt(J_1)+o(\expec{\opt(J_1)})\\
                &= \alpha(1+\delta)\opt(J_1)+o(\expec{\opt(J_1)}).
\end{align*}
Using \cref{lem:decomposition} with $t=n$, we can see that $\opt(J_1)\le (1+2\delta)\expec{\opt(J_1)}+o\left(\expec{\opt(J_1)}\right)$.
Moreover, we know that $\expec{\Opt(J_1)}=\expec{\Opt(J_2)}$.
Therefore, we can transform the above inequality to obtain the following form
for the upper bound on $\bluep(J_2,\mathcal A_\alpha, J_1)$.
\begin{align}
\bluep(J_2,\mathcal A_\alpha, J_1)&\le \alpha(1+\delta)(1+2\delta)\expec{\opt(J_1)}+o(\expec{\opt(J_1)})\nonumber\\
                &\le \alpha(1+4\delta)\expec{\opt(J_1)}+o(\expec{\opt(J_1)})\nonumber\\
                &= \alpha(1+4\delta)\expec{\opt(J_2)}+o(\expec{\opt(J_2)}).
\label{eq:final-blueprint}
\end{align}
In other words, using an $\alpha$-approximate offline algorithm to pack $J_1$,
we can pack $J_2$ in an \emph{online} manner while achieving (essentially)
the same approximation factor.

\subsection{Algorithm Assuming that the Value of \texorpdfstring{$n$}{n} is Known}\label{algorithm}
\label{alg-desc}
We now describe our algorithm which assumes the knowledge of the number of items.
The input $I$ is denoted by the list $x_1,x_2,\dots,x_n$.
We partition the entire input into {$m:=1/\delta^2$} stages
$\stage{0},\stage{1},\dots,\stage{m-1}$.
The zeroth stage $\stage{0}$, called the \emph{sampling stage}, contains the first
$\delta^2n$ items, i.e., $x_1,x_2,\dots,x_{\delta^2n}$. For $j\in[m-1]$,
the stage $\stage j$ contains the items with index starting from $j\delta^2n+1$
till $\min(n,(j+1)\delta^2n)$. In essence, $\stage 0$ contains
the first $\delta^2n$ items, $\stage 1$ contains the next $\delta^2n$ items,
$\stage 2$ contains the next $ \delta^2n$ items, and so on.
Note that the number of stages $m$ is a constant, and the last stage may contain fewer than $\delta^2n$ items.
In any stage $\stage j$, we denote the set of large items and small items
by $\stagelarge j$ and $\stagesmall j$, respectively.
Note that for any $j\in[m-2]$,
$\abs{\stage j}=\abs{\stage{j-1}}$ and since all the items
are sampled independently from the same distribution, we know that
the joint distributions of the sets $T_j$ and $T_{j-1}$ are the same.
Hence, to pack $T_j$, we can first pack $T_{j-1}$ using a good offline approximation
algorithm $\mathcal A_\alpha$ and use that packing as a blueprint.
The last stage $T_{m-1}$ may contain fewer than $\delta^2 n$ items. However, its size is only a small fraction
in comparison to the entire input; hence, we can show that it does not affect the final packing much.

The algorithm is as follows. First we pack $\stage 0$, the sampling stage, using \nextfit{}.
The sampling stage contains only a small but a constant fraction of the entire input set; hence it uses
only a few number of bins when compared to the final packing but at the same time provides a good
estimate of the underlying distribution.
If $\abs{L_0}$, the number of large items in the sampling stage, is at most $\param^3 \weight{\stage 0}$,
then we continue using \nextfit{} for the rest of the entire input too. 
{Intuitively, \nextfit{} performs well in this case as most of the items are small.}
Thus, from now on, let us assume that $\abs{L_0}>\param^3 \weight{T_0}$.
Consider an intermediate point when all the items in the stage $\stage{j-1}$ have arrived and the first element of stage $\stage j$
is about to arrive ($j\ge 1$). At this point, we compute the packing $\mathcal A_\alpha(T_{j-1})$
and use it as a blueprint to pack $T_j$, i.e., we run $\bluep(T_j, \mathcal A_\alpha, T_{j-1})$ (\cref{alg:blueprint})
to pack $T_j$.
In this way, we pack all the stages.
We call this algorithm $\ALG$.
\cref{algo} gives a more formal pseudocode of $\ALG$.
\begin{algorithm}
\caption{$\ALG (x_1,x_2,\dots,x_n)$: An algorithm for online bin packing
under the \iid{} model
assuming that the number of items $n$ is known before-hand}
\begin{algorithmic}[1]
\Require $I=\{x_1,x_2,...,x_n\}$ where each $x_i$ $(i\in[n]$ is sampled independently and identically.
\State Initialize $m:=\frac{1}{\delta^2}$ \Comment{Number of stages}
\For{$j$ in $\left\{0,1,\dots,m-1\right\}$}
    \State Define $\stage j$ to be the sequence $x_{j\delta^2n+1}, x_{j\delta^2n+2}, \dots, x_{\min\{n,(j+1)\delta^2n\}}$\Comment{The $j\Th$ stage}
\EndFor
\State Pack the sampling stage $T_0$ using \nextfit{}.
\If{$\abs{\stagelarge 0} \leq \param^3 \weight{\stage{0}}$}\label{line:less-large} \Comment{Very few large items in the sampling stage}
\State Use \nextfit{} for all the remaining stages $T_1, T_2, \dots, T_{m-1}$.
\Else
    \For{$j$= $1$ to $m-1$}
        \State Pack the stage $T_j$ using the blueprint packing procedure $\bluep(T_j,\mathcal A_\alpha, T_{j-1})$.
    \EndFor
\EndIf
\end{algorithmic}
\label{algo}
\end{algorithm}

We will proceed to analyze the algorithm $\ALG$.
We split the analysis into the following two cases
depending on the truth value of the if condition on \cref{line:less-large} of \cref{algo}: when
$\abs{\stagelarge 0} \leq \param^3\cdot \weight{\stage{0}}$ and when
$\abs{\stagelarge 0} > \param^3\cdot \weight{\stage{0}}$.
\subsubsection{\texorpdfstring{\textbf{Case 1:} $\abs{\stagelarge 0} \leq \param^3\cdot \weight{\stage{0}}$}{Case 1}}
Recall that in this case, we just continue with \nextfit{} for all the remaining items. To bound the \nextfit{} solution, we first consider the number of bins that contain at least one large item.
For this, we bound the value of $\abs{I_\ell}$.
Then we consider the bins that contain only small items
and bound this value in terms of weight of all items $\weight{I}$.
\begin{claim}
With probability at least $1-4\exp\left(-\frac{\delta}{6}\expec{\Opt(I)}^{1/3} \right)$,
for some positive constant $a$, we have that
\begin{align*}
\abs\largeitems \leq \param\cdot\weight{\stage 0} + a\expec{\Opt(I)}^{2/3}.
\end{align*}
\end{claim}
\begin{proof}
As the sampling stage contains $\param^2 n$ items, $\expec{\abs{\stagelarge 0}}=\param^2\expec{\abs\largeitems}$. 
Two applications of Lemma~\ref{large-item-conc} give us the following inequalities, we have 
\begin{align*}
\prob{\abs{\stagelarge 0} \leq {\param^2} \expec{\abs\largeitems} - \expec{\weight{I}}^{2/3}} &\leq 2\exp\left(-\frac{\param}{3}\expec{\weight{I}}^{1/3} \right)\text{ and}\\
\prob{\abs{\largeitems} \geq  \expec{\abs\largeitems} + \expec{\weight{I}}^{2/3}} &\leq 2\exp\left(-\frac{\param}{3}\expec{\weight{I}}^{1/3} \right).
\end{align*}
From the above inequalities we have that,
\begin{align*}
\abs{\largeitems}\leq \frac{1}{\param^2}\abs{\stagelarge 0} + \left(1+\frac{1}{\param^2}\right)\expec{\weight{I}}^{2/3}
\end{align*}
with probability at least $1-4\exp\left(-\frac{\param}{3}\expec{\weight{I}}^{1/3} \right)$.
We can use the inequalities $2\weight{I}\geq {\opt(I)} \geq {\weight{I}}$ and
$\abs{\stagelarge 0} \leq \param^3\cdot \weight{\stage{0}}$ to conclude the proof of this claim.
\end{proof}
The number of bins that contain at least one large item is upper bounded by $\abs{I_\ell}$.
Now we bound the number of bins that contain only small items. Note that \nextfit{} fills each such bin (with at most one possible exception)
up to a capacity at least $(1-\delta)$. So, the number of bins containing only small items is at most
$\frac{1}{1-\param}\weight{I}+1$. Thus,
with high probability,
\begin{align*}
\Nf(I) &\le\abs{\largeitems} + \frac{1}{(1-\param)}\weight{I} + 1\\ 
    &\leq \abs{\largeitems} + (1+2\param)\weight{I} + 1\\ 
    &\leq \param\cdot\weight{\stage 0} + (1+2\param)\weight{I} + a\expec{\opt(I)}^{2/3}+1
\end{align*}
{for some constant $a$.}

Using Lemma~\ref{weight-conc}, we get that, with high probability, $\weight{I}\leq \expec{\weight{I}} + \expec{\weight{I}}^{2/3}$. Using the facts $\weight{\stage 0}\le\weight{I}$ and ${\opt(I)}/2\leq{\weight{I}}\leq {\opt(I)}$, we get,
\begin{align}
\ALG(I)=\Nf(I)\leq (1+3\param)\expec{\opt(I)} + a_1\expec{\opt(I)}^{2/3},\label{continue-bf}
\end{align}
with high probability for some constant $a_1>0$.
\subsubsection{\texorpdfstring{\textbf{Case 2:} $\abs{\stagelarge 0} > \param^3\cdot \weight{\stage{0}}$}{Case 2}}
We split our analysis in this case into two parts. We first analyze the number of bins used in the sampling stage $\stage{0}$ and then analyze the number of bins used in the remaining stages.

Using \cref{large-item-conc}, we obtain w.h.p. that $\abs{\stagelarge 0}\le \delta^2\expec{\abs{\largeitems}}+\expec{\weight{I}}^{2/3}$.
Hence,
\begin{align}
    \expec{\abs{\largeitems}}&\ge\frac{1}{\delta^2}\abs{\stagelarge 0}-\frac{1}{\delta^2}\expec{\weight{I}}^{2/3}\nonumber\\
        &\ge\delta\weight{\stage 0}-\frac{1}{\delta^2}\expec{\weight{I}}^{2/3}.
        \label{some-rand-eq}
\end{align}
From \cref{weight-conc}, since $\abs{T_0}=\delta^2\abs{I}$, we obtain that the inequality $\weight{T_0}\ge \delta^2\expec{\weight{I}}-\expec{\weight{I}}^{2/3}$
holds with high probability. Substituting this inequality in \cref{some-rand-eq}, we obtain that
\begin{align}
\expec{\abs{I_\ell}}&\ge\param^3\expec{\weight{I}} - {\left(\delta+\frac{1}{\delta^2}\right)}\expec{\weight{I}}^{2/3}.
\label{somerandeq-1}
\end{align}
For any $j\ge1$, using the fact that ${\abs{\stage j}}= \param^2\abs{I}$ and using \cref{large-item-conc}, we obtain
the inequality $\abs{\stagelarge j}\ge \delta^2\expec{\abs{\largeitems}}-\expec{\weight{I}}^{2/3}$
that holds with high probability. Substituting \cref{somerandeq-1} in this inequality, we obtain, \whp{},
\begin{align}
    \abs{\stagelarge j} \geq \param^5\expec{\weight{I}} - {({2+\delta^3})}\expec{\weight{I}}^{2/3}.\label{some-rand-eq-2}
\end{align}

Each of the \cref{somerandeq-1,some-rand-eq-2} holds with high probability.

Note that $\weight{I}\ge \opt(I)/2$. So from now on, we assume that there exist constants $C_1, C_2>0$ which depend on $\delta$ such that, w.h.p., both the
following inequalities hold.
\begin{align}\label{expec-large}
    \expec{\abs{\largeitems}} &\geq C_1\cdot\expec{\opt(I)}. \\
\label{stage-large}
    \abs{\stagelarge j} &\geq C_2\cdot\expec{\opt(I)}.
\end{align}

\begin{itemize}
\item \textbf{Analysis of the Sampling Stage:} Recall that the number of items considered in the sampling stage is $ \sample$. We will bound the number of large items and the weight of items in this stage using Bernstein's inequality.
\begin{enumerate}
\item Since sampling stage has $\param^2n$ items, $\expec{\abs{\stagelarge 0}} = \param^2\expec{\abs{\largeitems}}$. By applying Bernstein's inequality
{for $X_1,X_2,\dots,X_{\abs{\stage 0}}$ where $X_i$ takes value $1$ is $x_i$ is large and $0$ otherwise, }we get,
\begin{align*}
\prob{\abs{\stagelarge 0} \geq 2\param^2\expec{\abs\largeitems}} & = \prob{\abs{\stagelarge 0} \geq\expec{\abs{\stagelarge 0}} + \param^2\expec{\abs\largeitems}} \\
& \leq 2\exp\left(-\frac{\param^4\expec{\abs\largeitems}^2}{2\expec{{\abs{\stagelarge 0}}} + \frac{2}{3}\param^2\expec{\abs\largeitems}}\right)\\
&\leq 2\exp\left(-\frac{1}{3}\param^2\expec{\abs\largeitems}\right) \\
& \leq 2\exp\left(-a_2\cdot\expec{\opt(I)}\right) \tag{from \cref{expec-large}}
\end{align*}
for some constant $a_2>0$.
So, with high probability, 
\begin{align}
\abs{\stagelarge 0}&\leq 2\param^2\expec{\abs{\largeitems}}\nonumber\\
    &\leq 2\param\expec{\opt(I)}.
    \label{not-rand-eq-1}
\end{align}
\item Similarly, $\expec{\weight{\stage 0}} = \param^2\expec{\weight{I}}$. By applying Bernstein's inequality
{for $X_1,X_2,\dots,X_{\abs{\stage 0}}$ where $X_i$ takes value $x_i$,} we get,
\begin{align*}
\prob{\weight{\stage 0} \geq 2\param^2\expec{\weight{I}}} & = \prob{\weight{\stage 0} \geq\expec{\weight{\stage 0}} + \param^2\expec{\weight{I}}} \\
& \leq 2\exp\left(\frac{-\param^4\expec{\weight{I}}^2}{2\param^2\expec{{\weight{\stage 0}}} + \frac{2}{3}\param^2\expec{\weight{I}}}\right) \\
& \leq 2\exp\left( \frac{-\param^2}{3}\expec{\weight{I}}\right)
\leq 2\exp\left(\frac{-\delta^2\expec{\opt(I)}}{6}\right).
\end{align*}
So, with high probability we have, 
\begin{align}
\weight{\stage 0} &\leq 2\param^2\expec{\weight{I}}\nonumber\\
    &\leq 2\param^2\expec{\opt(I)}.
    \label{not-rand-eq-2}
\end{align}
\end{enumerate}
Among the $\Nf(T_0)$ number of bins used by \nextfit{} to pack the sampling stage, the number of bins that contain
at least one large item is at most $\abs{L_0}$. On the other hand, each bin (with at most one exception) that contains only small items is filled up to
a volume of at least $(1-\delta)$. Hence,
\begin{align}
    \Nf(\stage 0) &\leq \abs{L_0}+\frac{1}{1-\delta}\weight{S_0}+1\nonumber\\
                &\leq \abs{L_0}+\frac{1}{1-\delta}\weight{T_0}+1.\nonumber
\end{align}
Substituting \cref{not-rand-eq-1,not-rand-eq-2} in the above inequality, we obtain that the following inequality holds with high probability.
\begin{align}
     \Nf(T_0)&\le2\param\expec{\opt(I)} + \frac{2\param^2}{1-\param}\expec{{\opt(I)}}\nonumber\\
                & \leq 4\param\expec{\opt(I)}.
    \label{sampling-stage}
\end{align}

\item \textbf{Analysis of the Remaining Stages:}
Recall that we pack each of the remaining stages using the blueprint packing $\mathcal A_\alpha(T_0)$.
The analysis of the last stage $T_{m-1}$ is slightly different as it is possible that $\abs{T_{m-1}}<\delta^2n=\abs{T_0}$.
Hence we analyze stages $T_1,T_2,\dots,T_{m-2}$ for now. At the end, we analyze the stage $T_{m-1}$.
Consider any $j\in[m-2]$.
\cref{lem:blueprint} tells us that $\bluep(T_j,\mathcal A_\alpha, T_{j-1})$,
the number of bins required to pack $T_j$, is bounded, \whp{}, as follows.
\begin{align*}
\bluep(T_j,\mathcal A_\alpha, T_{j-1})\le \alpha(1+4\delta)\expec{\opt(T_j)}+o\left(\expec{\opt(T_j)}\right).
\end{align*}
Now, we use \cref{lem:decomposition} to obtain the following inequality that holds with high probability.
\begin{align}
\bluep(T_j,\mathcal A_\alpha, T_{j-1})&\le \alpha(1+4\delta)(1+2\delta)\frac{\abs{T_j}}{n}\expec{\opt(I)}+o\left(\expec{\opt(T_j)}\right)\nonumber\\
                        &\le \alpha(1+4\delta)(1+2\delta)\frac{\abs{T_j}}{n}\expec{\opt(I)}+o\left(\expec{\opt(I)}\right)\nonumber\\
                        &\le \alpha(1+15\delta)\frac{\abs{T_j}}{n}\expec{\opt(I)}+o\left(\expec{\opt(I)}\right).\nonumber
\end{align}
Since the above inequality holds with high probability and since there are only constant number of stages $m$, by union bound,
the inequality obtained by summing the above inequality over all $j\in[m-2]$ also holds with high probability. In other words,
with high probability, we obtain that
\begin{align}
\sum_{j\in[m-2]}\bluep(T_j,\mathcal A_\alpha, T_{j-1})&\le\alpha(1+15\delta)\frac{\sum_{j=1}^{m-2}\abs{T_j}}{n}\expec{\opt(I)}+o\left(\expec{\opt(I)}\right)\nonumber\\
                    &\le\alpha(1+15\delta)\frac{\sum_{j=0}^{m-1}\abs{T_j}}{n}\expec{\opt(I)}+o\left(\expec{\opt(I)}\right)\nonumber\\
                    &=\alpha(1+15\delta)\expec{\opt(I)}+o\left(\expec{\opt(I)}\right).
\label{eq:rem-stages-except-last}
\end{align} 
Now, we analyze the last stage $T_{m-1}$. Recall that \cref{lem:blueprint} required the condition that
$\abs{J_1}=\abs{J_2}$. Since $T_{m-1}$ can have fewer items than stage $T_{m-2}$,
we cannot use \cref{lem:blueprint} directly. However, there is a simple workaround. Assume that, after the stage $T_{m-1}$ ends,
we sample $\delta n-\abs{T_{m-1}}$ \kvn{number of extra items.} Let $E$ denote this list of extra items. Let $T'_{m-1}$ be the list of items
$T_{m-1}$ appended with the list $E$. Now, we can use \cref{lem:blueprint} if we pack $T'_{m-1}$ 
using the blueprint packing $\mathcal A_\alpha(T_{m-1})$ since $\abs{T_{m-1}}=\abs{T'_{m-1}}$. Since $T_{m-1}$ is only a 
prefix of $T'_{m-1}$, we have that
\begin{align}
\bluep(T_{m-1},\mathcal A_\alpha,T_{m-2})&\le \bluep(T'_{m-1},\mathcal A_\alpha,T_{m-2})\nonumber\\
                            &\le \alpha(1+4\delta)\expec{\Opt(T'_{m-1})}+o(\expec{\Opt(T'_{m-1})})\nonumber\\
                            &= \alpha(1+4\delta)\expec{\Opt(T_{m-2})}+o(\expec{\Opt(T_{m-2})})\nonumber\\
                            &\le\alpha(1+4\delta)(1+2\delta)\frac{\abs{T_{m-2}}}{n}\expec{\Opt(I)}+o(\expec{\Opt(I)})\nonumber\tag{using \cref{lem:decomposition}}\\
                            &=\alpha(1+4\delta)(1+2\delta)\delta^2\expec{\Opt(I)}+o(\expec{\Opt(I)})\nonumber\\
                            &\le2\alpha\delta^2\expec{\Opt(I)}+o(\expec{\Opt(I)})
    \label{eq:last-stage-analyze}
\end{align}
Adding up \cref{eq:rem-stages-except-last} and \cref{eq:last-stage-analyze} completes the analysis of the remaining stages.
\begin{align}
\sum_{j\in[m-1]}\bluep(T_j,\mathcal A_\alpha, T_{j-1})&\le \alpha(1+15\delta+2\delta^2)\expec{\Opt(I)}+o(\expec{\Opt(I)})\nonumber\\
                                                &\le \alpha(1+16\delta)\expec{\Opt(I)}+o(\expec{\Opt(I)})
\label{eq:remaining-stages}
\end{align}

\end{itemize}
We now combine the analyses of the sampling stage and the remaining stages.
For the sampling stage, from \cref{sampling-stage}, we have $\Nf(\stage 0)\leq 4\param\expec{\opt(I)}$ with high probability. 
For all the remaining stages, \cref{eq:remaining-stages}, gives the upper bound on the number of bins used.
Combining both the results, we get an upper bound on $\ALG(I)$ that holds \whp{}
\begin{align}
\ALG(I)&\leq \alpha(1+20\param)\expec{\opt(I)}+ o(\expec{\opt(I)})\nonumber\\
          &  \leq \alpha(1+\eps)\expec{\opt(I)}+ o(\expec{\opt(I)}).\label{high-prob-bound}
\end{align}

In the low probability event when \cref{high-prob-bound} may not hold,
we can bound $\ALG(I)$ as follows. In the sampling stage,
we have that $\Nf(\stage 0)\le 2\opt(I)-1$.
For the remaining stages, we bound the number of bins containing at least one large item and the number of
bins containing only small items.
{Each large item is considered at most twice in the packing: first--when it arrives in the input list, second--when it is in the sampling stage and thus takes the role of a proxy item.}
So, the number of bins containing at least one large item is at most $2\abs{\largeitems}$.
In each stage, with one possible exception, every bin opened which has only small items has an occupancy of at least
$(1-\delta)$. Combining over all the stages, the number of bins which contain only small items is at most
$\frac{1}{1-\param}\weight{\smallitems}+m$. Thus,
we can bound the total number of bins used by $\ALG$ to be at most
$2\opt(I)+2\abs{\largeitems}+\frac{1}{1-\param}\weight{\smallitems}+m$.
On the other hand, we know that $\opt(I)\ge\weight{I}\ge\delta\abs{\largeitems}+\weight{\smallitems}$.
Hence, we obtain that $2\abs{\largeitems}+\frac{1}{1-\param}\weight{\smallitems}\le\frac{2}{\param(1-\param)}\opt(I)$.
Combining all these, we obtain that
\begin{align}
\ALG(I)\le\left(2+\frac{2}{\delta(1-\delta)}\right)\opt(I)+m\label{low-prob-bound}
\end{align}
Now, to obtain the competitive ratio, suppose \cref{high-prob-bound}
holds with probability $p$ $(=1-o(1))$. We combine \cref{high-prob-bound,low-prob-bound} similar to
the case when $\abs{\stagelarge 0} \le \param^3\cdot \weight{\stage{0}}$.
\begin{align*}
\expec{\ALG(I)}&\leq p\Big(\alpha(1+\eps)\expec{\opt(I)}+o(\expec{\opt(I)})\Big)\\
                &\qquad+(1-p)\left(\left(2+\frac{2}{\delta(1-\delta)}\right)\expec{\opt(I)}+m\right)\\
&\le\alpha(1+\epsilon)\expec{\opt(I)}+o(\expec{\opt(I)}).\tag{since $1-p=o(1)$ and $\delta\le \eps/16$}
\end{align*}
Scaling the initial value $\epsilon$ to $\epsilon/\alpha$ before the start of the algorithm,
we obtain a competitive ratio of $\alpha+\epsilon$.

\subsection{Getting Rid of the Assumption on the Knowledge of the Input Size}
\label{sec:get-rid}
In this subsection, we will extend $\ALG$
to devise an algorithm for online bin packing with \iid{} items
that guarantees essentially the same competitive
ratio as $\ALG$ without knowing the value of $n$. We denote this algorithm by $\IMPALG$.
\label{sec:imp-algorithm}

{Let $\mu:=\delta^2$}.
We first guess the value of $n$ to be a constant $\guess{0}:=1/\delta^3$.
Then, we run $\ALG$ until $\min\{n,\guess{0}\}$ items arrive (here, if $\min\{n,\guess{0}\}=n$,
then it means that the input stream has ended {before $n_0$ items have arrived}).
If $n>n_0$, i.e., if there are more items to arrive, then we revise our estimate of
$n$ by {multiplying it with $(1+\mu)$}, i.e., the new value of $n$ is set as
{$\guess{1}:=(1+\mu)\guess{0}$}. We start $\ALG$ afresh for the items with indices $n_0+1,n_0+2,\dots,\min\{n_1,n\}$.
{If $n>(1+\mu)n_0$, then we set the new guess of $n$ to be $n_2:=(1+\mu)n_1=(1+\mu)^2n_0$} and start $\ALG$
afresh on the items with indices $n_1+1,n_1+2,\dots,\min\{n_2,n\}$. We continue this process of {multiplying our estimate
of $n$ with $(1+\mu)$} until all the items arrive. See \cref{img:doubling} for an illustration.
The pseudocode is provided in \cref{algo-doubling}.
\begin{figure}[H]
\includesvg[width=\linewidth]{img/doubling}
\caption{
    {The division of input into super-stages to get rid of the assumption on the knowledge of $n$.
    The $(j+1)\Th$ super-stage is denoted by $\Gamma_j$.
    Super-stage $\Gamma_0$ contains $n_0=1/\delta^3$ items, $\Gamma_0\cup \Gamma_1$ contains $(1+\mu)n_0$ items,
    $\Gamma_0\cup \Gamma_1\cup \Gamma_2$ contains $(1+\mu)^2n_0$ items and so on.
    The last super-stage may not be full, but since it is very small in size compared
    to the entire input, it doesn't affect the performance of the algorithm.}
}
\label{img:doubling}
\end{figure}

\begin{algorithm}
\caption{$\IMPALG$: Improving $\ALG$ to get rid of the assumption on the knowledge of the number of items}
\begin{algorithmic}[1]
\State \textbf{Input:} $I_n(\mathcal{D})=\{x_1,x_2,...,x_n\}$.
\State $n_{-1}\leftarrow0;n_0\leftarrow\frac{1}{\delta^3}; \mathrm{for\:} j\ge 1, n_j={(1+\mu)}n_{j-1}$
\State $i\leftarrow 0$
\While{true}
    \State Run $\ALG\left(x_{n_{i-1}+1},x_{n_{i-1}+2},\dots,x_{\min\{n_i,n\}}\right)$ with one change that instead of packing the small items in the $S$-slots created for a stage, we maintain a global set of $S$-slots.
    \If{the input stream has ended}
        \State \textbf{return} the packing
    \Else
        \State $i\leftarrow i+1$
    \EndIf
\EndWhile
\end{algorithmic}
\label{algo-doubling}
\end{algorithm}

We consider the following partition of the entire input into super-stages as follows: The first super-stage, $\superstage{0}$,
contains the first $n_0$ items. The second super-stage, $\superstage{1}$, contains the next $n_1-n_0$
items. In general, for $i>0$, the $(i+1)\Th$ super-stage, $\superstage{i}$, contains $\min\{n_i,n\}-n_{i-1}$ items which are
given by $x_{n_{i-1}+1},x_{n_{i-1}+2},\dots,x_{\min\{n,n_i\}}$.
So, essentially, $\IMPALG$ can be thought of running $\ALG$ on each super-stage separately.
The number of super-stages is given by {$\kappa:=\ceil{\log_{(1+\mu)}(n/n_0)}$.}

{
    Note that $\abs{\superstage 0}=n_0$, $\abs{\superstage 1}=\mu n_0$, $\abs{\superstage 2}=\mu(1+\mu)n_0$ and so on.
    In general, for $j\in[\kappa-2]$, $\abs{\superstage j}=\mu(1+\mu)^{j-1}n_0$. Now consider the last super-stage $\superstage{\kappa-1}$
    and note that it may not be full, i.e., it can be the case that $\abs{\superstage{\kappa-1}}<\mu(1+\mu)^{\kappa-2}n_0$. So, $\ALG$
    may pack the last super-stage inefficiently.
    However, note that $\abs{\superstage{\kappa-1}}$ can be at most $\mu(1+\mu)^{\kappa-2}n_0\le \mu n$.
    Hence the last super-stage contains only a tiny fraction of the input and so, this will have very little effect on the
    packing of the entire input.
}

\subsubsection{Analysis}
When $n$ was known we only had $O(1)$ number of stages.
However, now we have $\kappa=\ceil{\log_{(1+\mu)}(n/n_0)}$ number of super-stages.
There can arise two problems:
\begin{itemize}
\item Recall that when $n$ was known, we derived a performance guarantee for each stage individually, that holds \whp{},
and used a union bound to combine all the stages.
However, here, since the number of super-stages is a super-constant,
we cannot hope for high probability guarantees using a union bound.
So, we consider the first $\kappa_1:=\ceil{\log_{(1+\mu)}(\delta^7 n)}$ 
number of the super-stages at a time. We show that these initial super-stages contain only a small fraction
of the entire input. Each of the final $(\kappa-\kappa_1)$ super-stages can be individually analyzed using the analysis of $\ALG$.
\item For each super-stage, we can have a constant number of $S$-bins (bins which contain only small items) with less occupancy. However, since
the number of super-stages itself is a super-constant, this can result in a lot of wasted space.
For this, we exploit the monotonicity of \nextfit{} to ensure that we can pack small items from a super-stage into empty slots for small items from the previous stages.
\end{itemize}
{We will now proceed to analyze $\IMPALG$.}
Recall that the number of super-stages
is given by {$\kappa=\ceil{\log_{(1+\mu)}(n/n_0)}$}
where $n_0$ was defined to be $1/\delta^3$. We will split the analysis into two parts. First,
we will analyze the number of bins used by our algorithm in the first {$\kappa_1=\ceil{\log_{(1+\mu)}(\delta^7 n)}$}
super-stages {as a whole}. Then, we will analyze
the final $\kappa_2:=\kappa-\kappa_1$ super-stages considering {each} one at a time.
We call the first $\kappa_1$ super-stages as \emph{initial super-stages} and the remaining
$\kappa_2$ super-stages as \emph{final super-stages}.

\noindent\textbf{Analysis of the initial super-stages: }The basic intuition of
the analysis of our algorithm in the initial
super-stages is as follows: Since only a small fraction of the entire input is present in these $\kappa_1$
super-stages, our algorithm uses only a small fraction of bins compared to the optimal packing of the
entire input. So, we bound the number of large items and the weight of small items and thus bound the
number of bins used.
\begin{lemma}
\label{small-fraction-input}
Let $\mathcal A$ be an algorithm for online bin packing such that in the packing output by $\mathcal A$,
every bin, except possibly a constant number of bins $\tau$,
which contains only small items has an occupancy of at least $(1-\delta)$.
Let $I$ be a sequence of $n$ \iid{} items and let $J$ be {any contiguous subsequence} of $I$ having size $\beta n$ $(0<\beta\le1)$.
Suppose $\mathcal A$ works in a way such that it packs at most $\nu$ copies of any large item where $\nu$
is a constant.\footnote{For example, $\ALG$ does this. {For each stage, it computes the proxy packing of the previous stage.
Hence, every large item is potentially packed twice.}
In the worst case, these copies may not get replaced,
thus resulting in wasted space.}
Then the following inequality holds with high probability.
\begin{align*}
\mathcal A(J)\le \frac{2\beta\nu}{\delta(1-\delta)}\expec{\opt(I)}+o(\expec{\opt(I)})
\end{align*}
\end{lemma}
\begin{proof}
Let $J_\ell,I_\ell$ respectively denote the set of large items in $J,I$ and let $J_s,I_s$ respectively denote the set of
small items in $J,I$. We prove the lemma considering two cases.

First, we consider the case when $\abs{J_\ell}\le\delta \opt(J)$. The number of bins in the packing of $J$ by $\mathcal A$
which contain at least one large item is upper bounded by $\nu\abs{J_\ell}$. The number of bins which contain only small items
is upper bounded by $\weight{J_s}/(1-\delta)+\tau$. Hence,
\begin{align*}
\mathcal A(J)&\le\nu\abs{J_\ell}+\frac{\weight{J_s}}{1-\delta}+\tau\\
			&\le\nu\delta\opt(J)+\frac{\opt(J)}{1-\delta}+\tau\\
			&=\frac{1+\nu(1-\delta)\delta}{1-\delta}\opt(J)+\tau\\
			&\le2\nu(1+\delta)\opt(J)+\tau\tag{since $0<\delta<1/2$ and $\nu\ge1$}\\
			&\le2\beta\nu(1+\delta)(1+2\delta)\expec{\opt(I)}+C_0\expec{\opt(I)}^{2/3}\qquad\textrm{w.h.p. for some constant $C_0$}\tag{using \cref{lem:decomposition}}\\
			&\le \frac{2\beta\nu}{\delta(1-\delta)}\expec{\opt(I)}+o(\expec{\opt(I)})\tag{since $\delta<1/2$, $1+\delta<\frac{1}{1-\delta}$ and $1+2\delta<\frac{1}{\delta}$}
\end{align*}
Next, we consider the case when $\abs{J_\ell}>\delta\opt(J)$. We will again bound the quantity $\nu\abs{J_\ell}+\frac{\weight{J_s}}{1-\delta}+\tau$.
To bound $\abs{J_\ell}$, we use Bernstein's inequality (\cref{bernstein}). Let $X_1,X_2,\dots,X_{\abs{J}}$ be random variables
where $X_i$ takes value $1$ if the $i\Th$ item in $J$ is large and $0$ otherwise. Then, clearly, $\abs{J_\ell}=\sum_{i=1}^{\abs J}X_j$.
Also, note that $\expec{\abs{J_\ell}}=\beta\expec{\abs{I_\ell}}$. Moreover,
we can derive the following inequalities that hold with high probability.
\begin{align*}
\expec{\abs{I_\ell}}&\ge\frac{1}{\beta}\abs{J_\ell}-\frac{1}{\beta}\expec{\weight{I}}^{2/3}\tag{using \cref{large-item-conc}}\\
		&\ge\frac{\delta}{\beta}\opt(J)-\frac{1}{\beta}\expec{\weight{I}}^{2/3}\\
		&\ge\frac{\delta}{\beta}\weight{J}-\frac{1}{\beta}\expec{\weight{I}}^{2/3}\\
		&\ge\frac{\delta}{\beta}\left(\beta\expec{\weight{I}}-\expec{\weight{I}}^{2/3}\right)-\frac{1}{\beta}\expec{\weight{I}}^{2/3}\tag{using \cref{weight-conc}}\\
		&\ge\frac{\delta}{2}\expec{\opt(I)}-o(\expec{\opt(I)})\tag{since $\weight{I}\le\opt(I)\le2\weight{I}$}
\end{align*}
Hence, from now on, we will assume that there exists a constant $a_1>0$ such that
\begin{align}
\expec{\abs{I_\ell}}\ge a_1\expec{\opt(I)}\label{some-random-eq}
\end{align}
{holds with high probability.}
Using Bernstein's inequality (\cref{bernstein}),
\begin{align*}
\prob{\abs{J_\ell} \geq 2\beta\expec{\abs\largeitems}} & = \prob{\abs{J_\ell} \geq\expec{\abs{J_\ell}} + \beta\expec{I_\ell}} \\
& \leq 2\exp\left(-\frac{\beta^2\expec{\abs\largeitems}^2}{2\expec{{\abs{J_\ell}}} + \frac{2}{3}\beta\expec{\abs\largeitems}}\right)\\
&= 2\exp\left(-\frac{3}{8}\beta\expec{\abs\largeitems}\right) \\
& \leq 2\exp\left(-\frac{3}{8}\beta a_1\cdot\expec{\opt(I)}\right) \tag{from \cref{some-random-eq}}
\end{align*}
Now to bound $\weight{J_s}$, we will again use Bernstein's inequality on a new set of random variables $X_1,X_2,\dots,X_{\abs{J}}$
where $X_i$ equals the weight of the $i\Th$ item in $J$ if it is small and $0$ otherwise. Clearly, $\sum_{i=1}^{\abs J}X_i=\weight{J_s}$
and $\expec{\weight{J_s}}=\beta\expec{\weight{I_s}}\le\beta\expec{\weight{I}}$. Thus,
\begin{align*}
\prob{\weight{J_s} \geq 2\beta\expec{\weight{I}}}  &\leq \prob{\weight{J_s} \geq\expec{\weight{J_s}} + \beta\expec{\weight{I}}} \\
& \leq 2\exp\left(\frac{-\beta^2\expec{\weight{I}}^2}{2\expec{{\weight{J_s}}} + \frac{2}{3}\beta\expec{\weight{I}}}\right)\\
 &= 2\exp\left( \frac{-3\beta}{8}\expec{\weight{I}}\right)\\
 &\leq 2\exp\left(\frac{-3\beta\expec{\opt(I)}}{16}\right)\tag{since $2\weight{I}\ge\opt(I)$}
\end{align*}
Thus, with high probability, we have that $\abs{J_\ell}\le2\beta\expec{\abs{I_\ell}}$ and $\weight{J_s}\le2\beta\expec{\weight{I}}$.
Hence,
\begin{align*}
\mathcal A(J)&\le\nu\abs{J_\ell}+\frac{\weight{J_s}}{1-\delta}+\tau\\
			&\le2\nu\beta\expec{\abs{I_\ell}}+\frac{2\beta\expec{\weight{I}}}{1-\delta}+\tau\\
			&\le\frac{2\nu\beta}{\delta}\expec{\opt(I)}+\frac{2\beta\expec{\opt(I)}}{1-\delta}+\tau\\
			&= \frac{2\nu\beta}{\delta(1-\delta)}\expec{\opt(I)}+o(\expec{\opt(I)})
\end{align*}
This completes the proof.
\end{proof}
With the help of the above lemma, we proceed to analyze $\IMPALG$ for the initial (first $\kappa_1$) super-stages.
Note that the number of items in the initial super-stages is given by $(1+\eps)^{\kappa_1-1}n_0\le(\delta^7 n)n_0=\delta^4n$.
Since $\ALG$ satisfies the properties of $\mathcal A$ in \cref{small-fraction-input} and $\IMPALG$ just applies $\ALG$
multiple times, we can use \cref{small-fraction-input} to analyze $\IMPALG$.

There is one difficulty though. Let's call a bin which contains only small items to be an $S$-bin.
Although in a super-stage the number of $S$-bins not filled up to a level of at least $(1-\delta)$
is a constant, the number of initial super-stages itself is not a constant. We can work around this problem
by continuing {(Next-Fit)} NF to pack the small items in the $S$-slots created during the previous stages and the previous
super-stages as well. In other words, instead of packing the small items of a stage in $S$-slots created
only during that stage, we keep a global set of $S$-slots and pack the small items {in them} using NF. We now have to show that
our algorithm doesn't increase the number of bins with this change. Since the way in which the large items are packed
hasn't changed, the number of bins that contain large items does not change. Since NF is monotone
(even with varying bin sizes, \cite{DBLP:journals/dam/Murgolo88}), the number of bins containing only the small items will either
decrease or stays the same.

Using \cref{small-fraction-input}
with {$\beta=\delta^4$ and $\nu=2$ (since in each super-stage, a large item is packed at most twice -- once when it arrives as a real item, and once when it is used as a proxy item)},
we obtain that with high probability,
{
\begin{align}
\ALG(\Gamma_0)+\ALG(\Gamma_1)+\dots+\ALG(\Gamma_{\kappa_1-1})
	&\le\frac{2(2)\delta^4}{\delta(1-\delta)}\expec{\opt(I)}+o(\expec{\opt(I)})\nonumber\\
	&\le8\delta^3\expec{\opt(I)}+o(\expec{\opt(I)})\nonumber\\
	&\le8\delta\expec{\opt(I)}+o(\expec{\opt(I)})
\label{initial-super-stages}
\end{align}
}
\noindent\textbf{Analysis of the final super-stages: }For this part, the important thing
to note is that {$\kappa_2\le\ceil{\log_{(1+\mu)}(1/(\delta^7 n_0))}$} is a constant. Moreover,
since the number of items in the initial super-stages is at least {$\delta^4n/(1+\mu)$},
each of the final super-stages has at least {$\mu\delta^4n/(1+\mu)$} items
(which tends to infinity in the limiting case). Thus, we can use the analysis of $\ALG$
for each of these final super-stages.
As mentioned, since the last super-stage
might not be full, we analyze the last super-stage differently.
All the other super-stages are full. So, we can directly use the analysis of $\ALG$.
For all $\kappa_1\le i<\kappa-1$,
from the analysis of $\ALG$ (\cref{high-prob-bound}), we have that with high probability
\begin{align*}
\ALG(\Gamma_i)\le\alpha(1+8\delta)\expec{\opt(\Gamma_{i})}+C\expec{\opt(\Gamma_{i})}^{2/3}+o(\opt(\Gamma_{i}))
\end{align*}
for some constant $C$.

Now, to bound $\opt(\Gamma_{i})$ in terms of $\opt(I)$, we can use \cref{lem:decomposition}. Thus, the above
inequality transforms into
\begin{align*}
\ALG(\Gamma_i)&\le\alpha(1+8\delta)(1+2\delta)\frac{\abs{\Gamma_{i}}}{n}\expec{\opt(I)}+C_1\expec{\opt(\Gamma_i)}^{2/3}+o(\opt(\Gamma_i))\\
&\le\alpha(1+8\delta)(1+2\delta)\frac{\abs{\Gamma_{i}}}{n}\expec{\opt(I)}+C_1\expec{\opt(I)}^{2/3}+o(\opt(I))
\end{align*}
for some constant $C_1$. Summing over all $i$ $(\kappa_1\le i<\kappa-1)$, and observing that $\kappa-\kappa_1$
is a constant, we obtain that with high probability,
\begin{align}
&\ALG(\Gamma_{\kappa_1})+\ALG(\Gamma_{\kappa_1+1})+\dots+\ALG(\Gamma_{\kappa-2})\nonumber\\
\le&\alpha(1+8\delta)(1+2\delta)\sum_{i=\kappa_1}^{\kappa-2}\frac{\abs{\Gamma_{i}}}{n}\expec{\opt(I)}+C_2\expec{\opt(I)}^{2/3}+o(\opt(I))\nonumber\\
\le&\alpha(1+26\delta)\sum_{i=\kappa_1}^{\kappa-2}\frac{\abs{\Gamma_{i}}}{n}\expec{\opt(I)}+C_2\expec{\opt(I)}^{2/3}+o(\opt(I))\nonumber\\
\le&\alpha(1+26\delta)\expec{\opt(I)}+C_2\expec{\opt(I)}^{2/3}+o(\opt(I))
\label{middle-super-stages}
\end{align}
for some constant $C_2$.

Now consider the last super-stage $\Gamma_{\kappa-1}$. If $\abs{\Gamma_{\kappa-1}}$ is exactly equal to $n_{\kappa-1}-n_{\kappa-2}$, then we can obtain the
same bound as \cref{middle-super-stages}.
However, this might not be the case.
{
	The last super-stage contains $n-n_{\kappa-2}$ items.
	\begin{align*}
		n-n_{\kappa-2}\le n_{\kappa-1}-n_{\kappa-2}=(1+\mu)^{\kappa-1}n_0-(1+\mu)^{\kappa-2}n_0
								&=\mu(1+\mu)^{\kappa-2}n_0\\
								&=\mu n_{\kappa-2}\\
								&\le\mu n
	\end{align*}
Hence, the size of the last super-stage is at most $\mu$ fraction of the entire input size.
We will again use \cref{small-fraction-input} with $\beta=\mu$ and $\nu=2$ for the
last super-stage. We thus obtain that
	\begin{align}
		\ALG(\superstage{\kappa-1})&\le \frac{(2)(2)(\abs{\superstage{\kappa-1}}/n)}{n\delta(1-\delta)}\expec{\opt(I)}+o(\expec{\opt(I)})\nonumber\\
							&\le \frac{4\mu}{\delta(1-\delta)}\expec{\opt(I)}+o(\expec{\opt(I)})\nonumber\\
							&\le 8\delta\expec{\opt(I)}+o(\expec{\opt(I)})\label{last-super-stage}
	\end{align}
	with high probability. The last inequality follows since $\mu=\delta^2$ and $\delta<1/2$.
}

Summing \cref{initial-super-stages,middle-super-stages,last-super-stage}, we obtain that
with high probability, for some constant $C_4$,
{
\begin{align}
\IMPALG\left(I\right)&=\sum_{j=0}^{\kappa_1-1}\ALG(\Gamma_j)+\sum_{j=\kappa_1}^{\kappa-2}\ALG(\Gamma_j)+\ALG(\Gamma_{\kappa-1})\nonumber\\
&\le (\alpha(1+26\delta)+16\delta)\expec{\opt(I)}+C_2\expec{\opt(I)}^{2/3}+o(\expec{\opt(I)})+o(\opt(I))\nonumber\\
&\le \alpha(1+42\delta)\expec{\opt(I)}+C_2\expec{\opt(I)}^{2/3}+o(\expec{\opt(I)})+o(\opt(I))
\label{absolute-final-eq}
\end{align}
}

\cref{absolute-final-eq} holds with high probability, say $p=1-o(1)$.
In the scenario where the low probability event occurs, we can bound the number of bins used by
$\IMPALG$ using \cref{small-fraction-input} with $\beta=1$ and {$\nu=2$:
\begin{align*}
\IMPALG(I)\le\frac{2\cdot 2}{\delta(1-\delta)}\expec{\opt(I)}+o(\expec{\opt(I)})
\end{align*}
}
Let $E$ be the event when \cref{absolute-final-eq} holds. Then
\begin{align*}
\expec{\IMPALG(I)}&=\expec{\IMPALG(I)|E}\prob{E}+\expec{\IMPALG(I)|\overline E}\prob{\overline E}\\
		&\le\Big(\left(\alpha\left(1+42\delta\right)\right)\expec{\opt(I)}+o(\expec{\opt(I)})\Big)(1-o(1))\\
			&\:\:\:\:\:\:+\left(\frac{4}{\delta(1-\delta)}\expec{\opt(I)}+1\right)o(1)\\
			&=\left(\alpha\left(1+42\delta\right)\right)\expec{\opt(I)}+o(\expec{\opt(I)}
\end{align*}
Choosing $\delta=\frac{\eps}{42\alpha}$ ensures that the competitive ratio of $\IMPALG$ is $\alpha+\eps$.

\section{\bestfit{} under the Random-Order Model}
In this section, we will prove \cref{thm:bfroa,thm:gt1by3}.
To recall, \cref{thm:bfroa} shows that \bestfit{} achieves a random-order ratio of exactly $1$
when all the item sizes are more than $1/3$, while \cref{thm:gt1by3} shows a random-order ratio of
$1.49107$ when the item sizes are in the range $(1/4,1/2]$.
\subsection{Performance of \bestfit{} when Item Sizes are Larger than \texorpdfstring{$1/3$}{1/3}}
\label{bfgt13}
Albers et al.~\cite{albers_et_al_MFCS} showed that the asymptotic random-order ratio of the \bestfit{} algorithm is
at most $1.25$ when all the item sizes are more than $1/3$. In this section, we improve it further and show that, \bestfit{} for this special case under the random-order model is nearly optimal.
We will use the upright matching result given by \cref{lem:upright-matching-roa}.
We restate the lemma again for convenience.
\uprightmatchingroa*
\begin{remark}
    \label{rem:upright-matching-roa}
    The above lemma can be looked at as follows. Consider an upright matching instance where the minus points are given by the set
    $P^-=\{(i,r_i)\}_{i\in[m]}$ and the plus points are given by the set $P^+=\{(m+i,s_i)\}_{i\in[m]}$. Now, if we randomly permute the
    $x$-coordinates of the points in $P^+\cup P^-$, then the number of unmatched points in a maximum upright matching is very low.
\end{remark}

We now proceed to prove \cref{thm:gt1by3}.
We first show that the Modified Best-Fit algorithm~\cite{coffman1993probabilistic} is nearly optimal using \cref{lem:upright-matching-roa}. 
The Modified Best-Fit (MBF) algorithm is the same as BF except that it closes a bin if it receives an item of size less than $1/2$. 
Hence, note that any bin packed by MBF can have at most two items.
Shor\cite{shor1986average} showed that MBF \emph{dominates} BF, i.e., for any instance $I$, $\Bf(I)\leq \mbf(I)$. 
MBF can be easily reduced to upright matching as follows.
Consider any bin packing instance $I=\{x_1,x_2,\dots,x_n\}$ and an arbitrary item $x_i\in I$.
\begin{itemize}
    \item If $x_i\le 1/2$, then we create a plus point $(i,1-x_i)$.
    \item If $x_i\ge 1/2$, then we create a minus point $(i,x_i)$.
\end{itemize}
Now, it can be seen that the maximum upright matching algorithm given by \cref{alg:uprightmatching} for the above instance exactly corresponds to
the MBF algorithm on $I$.
This is because any plus point $p^+$ corresponds to an item $x$ of size at most $1/2$ and any minus point $p^-$ corresponds to an item $y$
of size more than $1/2$, and $p^+$ is matched to $p^-$ by \cref{alg:uprightmatching} iff the item $x$ is packed on top of $y$
by MBF.

Call an item $x$ large ($L$) if $x>1/2$ and medium ($M$) if $x\in (1/3, 1/2]$. We define a bin as $LM$-bin if it contains one large item and one medium item. We use the following lemma which was proved in~\cite{albers_et_al_MFCS} using the monotonicity property of BF when all item sizes are more than $1/3$.
\begin{lemma}\label{LM-bins} \cite{albers_et_al_MFCS}
Consider the set of bin packing instances $\mathcal J$ where each instance $J\in\mathcal J$ only has large and medium items, and every bin in $\Opt(J)$ is an $LM$-bin.
If \bestfit{} has an AAR of $\alpha$ when restricted to the instances in $\mathcal J$, then it has an AAR of $\alpha$ for any list of items larger than $1/3$ as well.
\end{lemma}

Consider an input instance which has an optimal packing containing only $LM$-bins.
Consider the number of bins opened by MBF for such instances. Each large item definitely opens a new bin, and a medium item opens a new bin if and only if it
can not be placed along with a large item, i.e., it is ``unmatched''. So, the number of bins opened by MBF equals
(number of large items$+$number of unmatched medium items).
Now, we will prove our result.
\begin{theorem}
For any list $I$ of items larger than $1/3$, the asymptotic random order ratio $RR_{BF}^{\infty}=1$.
\end{theorem}
\begin{proof}
From Lemma~\ref{LM-bins}, it is enough to prove the theorem for any list $I$ for which $\opt(I)$ is only made up of $LM$-bins. So, we can assume that $I$ has $k$ large items and $k$ medium items where $\opt(I)=k$. Now consider the packing of $\mbf$ for a randomly permuted list $I_\sigma$.
We have, \[\mbf(I_{\sigma}) = (k + \text{number of unmatched medium items}).\]
We have already seen that MBF can be reduced to the maximum upright matching algorithm given by \cref{alg:uprightmatching}.
This, coupled with \cref{rem:upright-matching-roa}, tells us that the number of unmatched medium items is at most $O\left(\sqrt k\log^{3/4}k\right)$.
So, we have
\begin{align*}
\mbf(I_{\sigma}) &\leq k + O\left(\sqrt{k}(\log k)^{3/4}\right)\\
                &=\Opt(I)+o\left(\Opt(I)\right)
\end{align*}
with probability of at least $1-C\exp(-a(\log \opt(I))^{3/2})$ for some universal constants $a, C, K >0$. Since $\mbf$ dominates $\Bf$, we have
\[\prob{\Bf(I_{\sigma}) \leq \Opt(I)+o\left(\Opt(I)\right)} \geq 1-C\exp\left(-a\log^{3/2} \opt(I)\right).\]
In case the high probability event does not occur, we can use the bound of $\Bf(I_\sigma)\leq 1.7\opt(I)+2$.
Let $p:=C\exp(-a(\log \opt(I))^{3/2})$. Then
\begin{align*}
\expec{\Bf(I_{\sigma})}&\le p(1.7\expec{\opt(I)}+2)+(1-p)(\expec{\opt(I)}+o(\expec{\opt(I)}))\\
				&\le\expec{\opt(I)}+o(\expec{\opt(I)}).\tag{since $p=o(1)$}
\end{align*}
So, we get
\[
RR_{BF}^{\infty} = \limsup\limits_{k\rightarrow\infty}\left(\sup\limits_{I:\OPT(I)=k}(\mathbb{E}[\Bf(I_{\sigma})]/\OPT(I))\right) = 1.
\]
This completes the proof.
\end{proof}

\subsection{The $3$-Partition Problem under Random-Order Model}
\label{sec:roa}
In this section, we analyze the \bestfit{} algorithm under the random-order
model given that the item sizes lie in the range $(1/4,1/2]$, and thus prove Theorem \ref{thm:bfroa}.
We call an item \emph{small} if its size lies in the range $(1/4,1/3]$
and \emph{medium} if its size lies in the range $(1/3,1/2]$.
Let $I$ be the input list of items and let $n:=\abs I$.
Recall that given $\sigma$,
a uniform random permutation of $[n]$, $I_\sigma$ denotes the
list $I$ permuted according to $\sigma$.
We denote by $\opt(I_\sigma)$, the
number of bins used in the optimal packing of $I_\sigma$ and by $\Bf(I_\sigma)$,
the number of bins used by \bestfit{} to pack $I_\sigma$.
Note that $\opt(I_\sigma)=\opt(I)$.

If there exists a set of three small items in $I_\sigma$ such that they arrive
as three consecutive items, we call that set to be an \emph{$S$-triplet}.
We call a bin to be a $k$-bin if it contains
exactly $k$ items, for $k \in \{1,2,3\}$.
We sometimes refer to a bin by mentioning its contents more specifically as
follows: An $MS$-bin is a $2$-bin which contains a medium item and a small item.
Similarly, an $SSS$-bin is a $3$-bin which contains three small items.
Likewise, we can define an $M$-bin, $S$-bin, $MM$-bin, $SS$-bin, $MMS$-bin, and $MSS$-bin.

Since the item sizes lie in $(1/4,1/2]$, any bin in the optimal packing contains at most
three items. For the same reason, in the packing by \bestfit{}, every bin (with one possible exception)
contains at least two items. This trivially shows that the ECR of \bestfit{} is at most $3/2$.
To break the barrier of $3/2$, we use the following observations.
\begin{itemize}
    \item Any $3$-bin must contain a small item.
    \item So, if the optimal solution contains a lot of $3$-bins, then it means that the input set contains a
    lot of small items.
\end{itemize}

We will prove that if there exist many small items in the input, then with high probability,
in a random permutation of the input, there exist many disjoint $S$-triplets.
\begin{lemma}
\label{no_of_consecutive_triplets}
Let $m$ be the number of small items in the input set $I$,
and let $X_\sigma$ denote the maximum number of mutually disjoint $S$-triplets in $I_\sigma$.
Suppose $m\ge cn$ where {$c$ is a positive constant}, then the following statements hold true:
\begin{enumerate}
\item $\expec{X_\sigma} \ge m^3/(3n^2){\ge c^3n/3}$.
\item $X_\sigma\ge {c^3n/3-o(n)}$ with high probability.
\end{enumerate}
\end{lemma}
Then, we prove that \bestfit{} packs at least one small item from an $S$-triplet in a $3$-bin or
in an $SS$-bin.
\begin{lemma}
\label{triplets}
In the input sequence $I_\sigma$, let $\kappa$ denote the number of disjoint $S$-triplets.
Then, in the \bestfit{} packing of $I_\sigma$, there will be at least $\floor{\kappa/2}$
number of $3$-bins.
\end{lemma}
But the number of $SS$-bins in the final packing of \bestfit{} can be at most one.
So, we obtain that the number of $3$-bins in the \bestfit{} packing is significant.
With these arguments, the proof of \cref{thm:bfroa} follows. 
Now, we give the detailed proofs of the above two lemmas now.

\begin{proof}[Proof of \cref{no_of_consecutive_triplets}]
One can construct a random permutation of the input list
$I$ by first placing the small items in a
random order and then inserting the remaining items among the small items randomly.

After placing the small items, we have $m+1$ gaps to place the remaining items as shown below in the form of empty squares.
\begin{align*}
    \square \underbrace{S \square S\square S}_{\mathrm{Triplet}\:T_1}\square\underbrace{S \square S\square S}_{\mathrm{Triplet}\:T_2}\square\underbrace{S \square S\square S}_{\mathrm{Triplet}\:T_3}\cdots\underbrace{S \square S\square S}_{\mathrm{Triplet}\:T_{m/3}}\square
\end{align*}
As shown above, we name the $S$-triplets as $T_1,T_2,\dots,T_{m/3}$.
(Strictly speaking, the number of $S$-triplets should be $\floor{m/3}$,
but we relax this since we are only interested in the asymptotic case and  $m \ge cn$.)
Let us start inserting the medium items into these empty squares.
Consider any $S$-triplet $T_i$. The probability that $T_i$
continues to be an $S$-triplet after inserting the first medium
item is $(m-1)/(m+1)$. This is because among the $m+1$ squares, we have $m-1$ choices of squares
for the first medium item to be inserted in. The first medium item thus occupies one of the squares,
but in this process, it creates two new squares on either side of itself, thereby
increasing the net number of squares by one. Hence, the probability that $T_i$ continues
to be an $S$-triplet after inserting the second medium item (conditioned on it being a consecutive $S$-triplet after inserting the
first medium item) is $m/(m+2)$.
We continue this process of inserting the medium items and after they all have been inserted,
let $Y_\sigma^{(i)}$ be the indicator random variable denoting whether $T_i$ is consecutive or not. Then
\begin{align*}
    \expec{Y_\sigma^{(i)}}=\frac{m-1}{m+1} \frac{m}{m+2} \cdots \frac{n-2}{n}=\frac{m(m-1)}{n(n-1)}.
\end{align*}
Let $Y_\sigma=Y_\sigma^{(1)}+Y_\sigma^{(2)}+\dots+Y_\sigma^{(m/3)}$.
It can be seen that $Y_\sigma$ is a lower bound on $X_\sigma$ which was defined as the 
maximum number of mutually disjoint $S$-triplets in $I_\sigma$ in the lemma statement.
By linearity of expectations,
\begin{align*}
    \lim_{n,m\to\infty}\expec{Y_\sigma}=\lim_{n,m\to\infty}\frac{m}{3}\frac{m(m-1)}{n(n-1)}{\approx\frac{m^3}{3n^2}\ge\frac{c^3n}{3}}.
\end{align*}
Since $X_\sigma\ge Y_\sigma$, the first part of the lemma follows.

To prove the second part of the lemma, we will compute $\var{Y_\sigma}$
and use Chebyshev's inequality to show that $Y_\sigma$ is concentrated
around its mean. For any $i$, note that 
\begin{align}
\var{Y_\sigma^{(i)}}=\expec{Y_\sigma^{(i)}}-\expec{Y_\sigma^{(i)}}^2<1.
\label{variance}
\end{align}
This is because $Y_\sigma^{(i)}$ has expectation at most $1$.

Since $Y_\sigma$ is a sum of $Y_\sigma^{(i)}$-s, to calculate the variance of $Y_\sigma$, 
we also need to calculate
the covariance of each pair of random variables in the set $\left\{ Y_\sigma^{(1)},Y_\sigma^{(2)},\dots,Y_\sigma^{(m/3)}\right\}$.
Intuitively, it is clear that for $j\ne k$, the random variables 
$Y_\sigma^{(j)}$ and $Y_\sigma^{(k)}$ are negatively correlated.
This is because if the $S$-triplet $T_j$ remains to be an $S$-triplet after inserting all the small
items, then the probability of the medium items being inserted in between the gaps of the small items
in the $S$-triplet $T_k$ increase.
Hence, the chances that the $S$-triplet $T_k$ remains to be an $S$-triplet decrease.
We formalize this intuition now.

Consider any two $S$-triplets $T_j,T_k$ where $j,k\in[m/3]$ and $j\ne k$. The probability that both $T_j$
and $T_k$ remain to be $S$-triplets after inserting the first medium item is $(m-3)/(m+1)$.
This is because, among the $m+1$ squares, the first medium item can be inserted in $m-3$
squares in order to ensure that the $S$-triplets $T_j,T_k$ continue to be $S$-triplets. 
Similarly, the probability that both $T_j$
and $T_k$ remain to be $S$-triplets after inserting the second medium item is $(m-2)/(m+2)$.
Continuing in this manner, after all the medium items have been inserted,
\[
    \expec{Y_\sigma^{(j)}Y_\sigma^{(k)}}=\prob{Y_\sigma^{(j)}=1\land Y_\sigma^{(k)}=1}
    =\frac{m-3}{m+1}\frac{m-2}{m+2}\cdots\frac{n-4}{n}=\frac{m(m-1)(m-2)(m-3)}{n(n-1)(n-2)(n-3)}.
\]

The covariance of $Y_\sigma^{(j)}, Y_\sigma^{(k)}$ is given by
\begin{align}
    \cov{Y_\sigma^{(j)}}{Y_\sigma^{(k)}}&=\expec{Y_\sigma^{(j)}Y_\sigma^{(k)}}
                                            -\expec{Y_\sigma^{(j)}}\expec{Y_\sigma^{(k)}}\nonumber\\
            &=\frac{m(m-1)(m-2)(m-3)}{n(n-1)(n-2)(n-3)}-\frac{m^2(m-1)^2}{n^2(n-1)^2}\nonumber\\
            &=\frac{m(m-1)}{n(n-1)}\left(\frac{(m-2)(m-3)}{(n-2)(n-3)}-\frac{m(m-1)}{n(n-1)}\right).\nonumber
\end{align}
For any reals $x,y$ such that $2<x<y$, we have the identity $\frac{x-2}{y-2}<\frac xy$. 
Hence, we obtain that
\begin{align}
    \cov{Y_\sigma^{(j)}}{Y_\sigma^{(k)}}<0.\label{covariance}
\end{align}
Combining all these, the variance of $Y_\sigma$ can be calculated as follows
\begin{align*}
\var{Y_\sigma}&=\sum_{i=1}^{m/3}\var{Y_\sigma^{(i)}}+2\sum_{1\le j<k\le m/3}\cov{Y_\sigma^{(j)}}{Y_\sigma^{(k)}}\\
            &\le\frac{m}{3}\tag{from \cref{variance,covariance}}\\
            &\le n
\end{align*}
Now using Chebyshev's inequality,
\begin{align*}
\prob{Y_\sigma\le\expec{Y_\sigma}-\left(\expec{Y_\sigma}\right)^{2/3}}
                &\le\prob{\abs{Y_\sigma-\expec{Y_\sigma}}\ge\left(\expec{Y_\sigma}\right)^{2/3}}\\
                &\le\frac{\var{Y_\sigma}}{\left(\expec{Y_\sigma}\right)^{4/3}}\\
                &\le\frac{n}{\frac{c^4}{3^{4/3}}n^{4/3}}\\
                &=O\left(\frac{1}{n^{1/3}}\right), \text{as $c$ is a constant.}
\end{align*}
Since $X_\sigma$ is the maximum number of disjoint $S$-triplets
and $Y_\sigma$ denotes the number of disjoint $S$-triplets among $T_1,T_2,\dots,T_{m/3}$
after inserting all the medium items, we have $X_\sigma\ge Y_\sigma\ge c^3n/3-o(n)$ with high probability.
\end{proof}

\begin{proof}[Proof of \cref{triplets}]
To prove the lemma, we will show that each $S$-triplet will result in the formation of
a $3$-bin or an $SS$-bin. A property of \bestfit{} is that, at any point in time, there can be at most
one bin of load at most $1/2$. An $SS$-bin has a load at most $1/2$; hence, every $SS$-bin (with
at most one exception) created by the $S$-triplets will transform into a $3$-bin.

Let $\{S_1,S_2,S_3\}$ be an $S$-triplet in $I_\sigma$ such that $S_3$ follows
$S_2$ which in turn follows $S_1$.
We consider three cases to show that this $S$-triplet will result in the formation of a $3$-bin or an $SS$-bin.
\begin{itemize}
    \item If $S_1$ is going to be packed in a $2$-bin. Then after packing $S_1$, this bin becomes a $3$-bin and hence the lemma holds.
    \item Suppose $S_1$ is going to be packed in a $1$-bin $B$. Then just before the arrival of $S_1$, $B$
    must have been the only $1$-bin and hence, after packing $S_1$, each bin is either a $2$-bin or
    a $3$-bin.
    \begin{itemize}
        \item Now, if one of $S_2, S_3$ is packed into a $2$-bin, then the bin becomes a $3$-bin and the lemma follows.
        \item Otherwise, both $S_2, S_3$ are packed into a single new $SS$-bin and hence the lemma follows.
    \end{itemize}

    \item Suppose $S_1$ is packed in a new bin. Then just prior to the arrival of $S_2$,
     every bin is either a $2$-bin or $3$-bin except the bin containing $S_1$.
     Thus, $S_2$ is either packed in a $2$-bin (thus becoming a $3$-bin) or packed
     in the bin containing $S_1$, resulting in an $SS$-bin.
\end{itemize}
Thus, we have shown that each $S$-triplet will result in the formation of a $3$-bin or an $SS$-bin.
At any point of time, the number of $SS$-bins can be at most one. Therefore, every $SS$-bin (except at most one) 
created by the $S$-triplets will end up as a $3$-bin due to the future items.

However, consider the following scenario where an $S$-triplet $\{S_1,S_2,S_3\}$ creates an $SS$-bin $B_1$ and 
a future $S$-triplet results in the formation of a $3$-bin $B_2$, and $B_1=B_2$. In case such a scenario occurs,
two $S$-triplets correspond to only one $3$-bin in the packing of \bestfit{}. But once this scenario occurs,
the bin $B_1$ ($=B_2$) will be closed, i,e., no more items will be packed in it.

Hence, we obtain that if there are $\kappa$ number of $S$-triplets in the input sequence $I_\sigma$,
there will be at least $\floor{\kappa/2}$ number of $3$-bins in the \bestfit{} packing of $I_\sigma$.
\end{proof}
Now, using \cref{no_of_consecutive_triplets,triplets}, we prove Theorem \ref{thm:bfroa}.
\begin{proof}[Proof of Theorem \ref{thm:bfroa}]
Let $X_3$ be the number of $3$-bins in the final \bestfit{} packing of $I_\sigma$.
The remaining $(n-3X_3)$ items are packed in $2$-bins and at most one $1$-bin. Therefore,
\begin{align*}
    \Bf(I_{\sigma})\le X_3 + \frac{n-3X_3}{2}+1 = \frac{n-X_3}{2}  + 1.
\end{align*}
Since any bin in the optimal solution can accommodate at most three items,
we have that $\opt(I)\ge n/3$. Hence,
\begin{equation}
    \Bf(I_{\sigma})\le \frac{3}{2}\left(1-\frac{X_3}{n}\right)\opt(I) + 1.
    \label{bf-3bins}
\end{equation}
Let $z_1$ $(\le 1),z_2$ and $z_3$ be the number of $1$-bins, $2$-bins, and $3$-bins in the optimal packing of $I_\sigma$, respectively.
Then, $\opt(I)=z_1+z_2+z_3$ and $n=z_1+2z_2+3z_3$. Note that any two items can fit in a bin.
Define the quantity $\mu\coloneqq z_2/\Opt(I)$, the fraction of $2$-bins in the optimal solution. We obtain
\begin{equation}
    \Bf(I_{\sigma}) \leq \frac{n+1}{2} = \frac{z_1+2z_2+3z_3+1}{2}\le \frac{3\opt(I)-z_2-2z_1+1}{2}
                \le \left(\frac{3}{2}-\frac{\mu}{2}\right)\opt(I)+1. \label{first-analysis}
\end{equation}
When $\mu$ is close to one we already obtain a competitive ratio very close
to $1$ due to the above inequality. When $\mu$ is significantly less than $1$, we analyze the \bestfit{} packing of $I_\sigma$ in a different way.
Let $n_s$ be the number of small items. Since a $3$-bin contains at least one small item, we know that,
$n_s \geq z_3=\Opt(I)-z_2-z_1=\Opt(I)-\mu\Opt(I)-z_1 = (1-\mu)\opt(I)-z_1$. 
On the other hand, we have that $n=z_1+2z_2+3z_3=\Opt(I)+z_2+2z_3=\Opt(I)+\mu\Opt(I)+2((1-\mu)\Opt(I)-1)\le (3-\mu)\Opt(I)$.
Thus, we have following two inequalities.
\begin{align*}
n_s\ge (1-\mu)\opt(I)-z_1\quad\text{and}\quad n\le (3-\mu)\Opt(I)
\end{align*}
Hence, the fraction of small items in the input sequence is given by
\begin{align*}
f_s\coloneqq\frac{n_s}{n}\ge \frac{(1-\mu)\opt(I)-z_1}{(3-\mu)\Opt(I)}
\end{align*}
Let $X_\sigma$ be the random variable which denotes the maximum number of $S$-triplets in the input sequence $I_\sigma$.
By \cref{no_of_consecutive_triplets},
\[
    X_\sigma \geq \frac{((1-\mu)\opt(I_\sigma)-z_1)^3}{3(3-\mu)^3\Opt(I)^3}n-o(n)= \frac{(1-\mu)^3}{3(3-\mu)^3}n-o(n) {\text{ w.h.p.}}, \text{ as } z_1\le 1.
\]
Recall that $X_3$ denotes the number of $3$-bins in the \bestfit{} packing of $I_\sigma$.
By \cref{triplets}
\begin{align}
    X_3 \ge \frac{X_\sigma}{2}-1 \geq \frac{(1-\mu)^3}{6(3-\mu)^3}n-o(n)
    \label{eq:3bins00}
\end{align}
with high probability.
Substituting \cref{eq:3bins00} in \cref{bf-3bins}, we get,
\begin{align}
    \Bf(I_{\sigma}) &\leq \frac{3}{2} \left(1-\frac{o(n)}{n}-\frac{(1-\mu)^3}{6(3-\mu)^3}\right) \opt(I_\sigma) + 1\nonumber\\
     &\leq \frac{3}{2} \left(1-\frac{(1-\mu)^3}{6(3-\mu)^3}\right) \opt(I_\sigma) +o(\opt(I_\sigma))
     \label{second-analysis}
\end{align}
with high probability.
\cref{first-analysis,second-analysis} are the result of two different ways of analyzing the same algorithm. Hence, the expected competitive ratio is at most
\begin{align*}
&\min\left\{\frac32-\frac\mu2, \frac32-\frac{(1-\mu)^3}{4(3-\mu)^3}\right\}\\
=&\frac32-\frac12\max\left\{\mu, \frac{(1-\mu)^3}{2(3-\mu)^3}\right\}
\end{align*}
In the domain $(0,1)$, the function $\mu$ is increasing and takes value zero at $\mu=0$, while the function $\frac{1-\mu}{3-\mu}$ is decreasing and takes
value zero at $\mu=1$. Hence, the maximum among both the functions is achieved when both the functions are equal, i.e., for a $\mu\in[0,1]$ satisfying
\begin{align*}
\mu=\frac{(1-\mu)^3}{2(3-\mu)^3}
\end{align*}
This happens when $\mu=0.017861$ (can be verified by substitution); the approximation ratio, in this case, is roughly $1.49107$.

To summarize, we proved that, with high probability, say $p=1-o(1)$, $\Bf(I_\sigma)\le 1.49107\opt(I)+o(\opt(I))$.
The fact that $\Bf(I_\sigma)\le 1.7\opt(I)+2$ always holds due to \cite{johnson1974worst}. Combining these, we get that
$\expec{\Bf(I)}\le p(1.49107\cdot\expec{\opt(I)}+o(\expec{\opt(I)}))+(1-p)(1.7\expec{\opt(I)}+2)$. Thus we
obtain that 
\[
    \expec{\opt(I)}\le 1.49107\cdot\expec{\opt(I)}+o(\expec{\opt(I)}).
\]
This completes the proof of \cref{thm:bfroa}.
\end{proof}

\section{Multidimensional Online Vector Packing Problem under \iid{} model}
\label{sec:ovp}
In this section, we design an algorithm for $d$-dimensional online vector
packing problem (d-OVP) where $d$ is a positive integer constant.
In this entire section, we abbreviate a tuple with $d$ entries as just \emph{tuple}.
A tuple $Y$ is represented as $\left(Y^{(1)},Y^{(2)},\dots,Y^{(d)}\right)$.
We define $Y^{\max}$ to be $\max\left\{Y^{(1)},Y^{(2)},\dots,Y^{(d)}\right\}$
In d-OVP, the input set $I$ consists of $n$ tuples $X_1,X_2,\dots,X_n$
which arrive in online fashion; we assume that $n$ is
known beforehand.
Each $X_i^{(j)}$ is sampled independently from a distribution $\mathcal D^{(j)}$.
The objective is to partition the $n$ tuples into a minimum number of bins such that
for any bin $B$ and any $j\in[d]$, the sum of all the $j\Th$ entries of all the tuples
in $B$ does not exceed one. Formally, we require that for any bin $B$ and any $j\in[d]$,$\sum_{x\in B}x^{(j)}\le 1$.

\subsection{Algorithm}
Given an offline $\alpha$-asymptotic approximation algorithm $\auxalgo$ for
bin packing, we obtain a $(d\alpha+\epsilon)$-competitive algorithm for d-OVP as follows: For every input tuple $X_i$, we round each
$X_i^{(j)}, j\in[d]$ to $X_i^{\max}$. After rounding, since all the tuples have
same values in each of the $d$ entries, we can treat each tuple $X_i$
as an one-dimensional item of size $X_i^{\max}$.

It is easy to see that each $X_i^{\max}$ is independently sampled from the same
distribution: Let $F^{(j)}$ be the cumulative distribution function (CDF) of $\mathcal D^{(j)}$.
Then the CDF, $F$, of $X_i^{\max}$ (for any $i\in[n]$) is given by
\begin{align*}
F(y)=\prob{X_i^{\max}\le y}&=\prod_{j=1}^d\prob{X_i^{(j)}\le y}\tag{By independence of $X_i^{(j)}$s}\\
				&=\prod_{j=1}^dF^{(j)}(y)
\end{align*}
Hence, the problem at hand reduces to solving an online bin packing problem where
items are independently sampled from a distribution whose CDF is given by the
function $\prod_{i=1}^dF^{(j)}$. So, we can use the algorithm from \cref{sec:iid-model}.
\subsection{Analysis}
Let $I$ denote the vector packing input instance and
let $\overline I$ denote the rounded up one-dimensional bin packing instance.
\begin{lemma}
\label{lem:ovp}
Let $\optv(I)$ denote the optimal number of bins used
to pack $I$. Then $\opt(\overline I)\le d\optv(I)$.
\end{lemma}
\begin{proof}
Consider any optimal packing of $I$. We show how to construct a feasible packing of $\overline I$
starting from the optimal packing of $I$. Consider any bin in the optimal packing of $I$
and let $B$ denote the set of tuples packed inside it. Let $\overline{B}$ denote the rounded-up 
instance of $B$. Also, for $j\in[d]$, let $B^{(j)}$ denote the set of tuples whose
$j\Th$ dimension has the largest weight, i.e.,
\begin{align*}
B^{(j)}=\left\{Y\in B: Y^{(j)}=Y^{\max}\right\}
\end{align*}
If any tuple belongs to $B^{(j)}$ as well as $B^{(k)}$ for some $j\ne k$,
we break these ties arbitrarily and assign it to one of $B^{(j)},B^{(k)}$.
For $j\in[d]$, let $\overline{B^{(j)}}$ denote the rounded instance of $B^{(j)}$.
We can pack $\overline{B}$ in at most $d$ bins as follows: For any $j\in[d]$,
we know that
\begin{align*}
\sum_{Y\in B^{(j)}}Y^{(j)}\le 1
\end{align*}
and since for all $Y\in B^{(j)}$, $Y^{\max}=Y^{(j)}$, it follows that
\begin{align*}
\sum_{Z\in \overline{B^{(j)}}}Z\le 1
\end{align*}
Hence, we can pack every $\overline{B^{(j)}}$ in one bin and the lemma follows.
\end{proof}

Now we are ready to prove our main theorem of this section.
\begin{theorem}
For any $\eps>0$ and a given polynomial-time $\alpha$-approximation algorithm for online bin packing, we can obtain a polynomial-time algorithm for d-OVP with an asymptotic approximation ratio of $(\alpha d+\epsilon d)$.
\end{theorem}
\begin{proof}
Let us denote our present algorithm by $\ALGV$. Then we get:
\begin{align}
 \ALGV(I) &=\ALG(\overline I) \label{ine1} \\
&\le (\alpha+\epsilon)\opt(\overline I)+o\left(\opt(\overline I)\right) \label{ine2}\\
	&\le (\alpha d+\epsilon d)\optv(I)+o(\optv(I)).\label{ine3}
\end{align}
Here, Equation \eqref{ine1} follows from the property of $\ALGV$, \eqref{ine2} follows from the results of \cref{sec:iid-model}, and \eqref{ine3} follows from Lemma 
\ref{lem:ovp}. This concludes the proof.
\end{proof}

\section{Conclusion}
We studied online bin packing under two stochastic settings, namely the
\iid{} model, and the random-order model. For the first setting, we devised
a meta-algorithm which takes any offline algorithm $\mathcal A_\alpha$
with an AAR of $\alpha$ (where
$\alpha$ can be any constant $\ge1$), and produces an online algorithm with
an ECR of $(\alpha+\eps)$.
This shows that online bin packing under the \iid{} model and offline
bin packing are almost equivalent.
Using any AFPTAS as $\mathcal A_\alpha$ results in an online algorithm with an
ECR of $(1+\eps)$ for any constant $\eps>0$.
An interesting question of theoretical importance is
to find whether achieving an ECR of $1$ is possible or not.
Another related open question is if we can settle online bin packing
under the random-order model as well.

Then, we studied the analysis of the well-known \bestfit{} algorithm under the random-order model.
First, we proved that the ARR of \bestfit{} is equal to one if all the item sizes are greater
than $1/3$. Then, we improved the analysis of the \bestfit{} from $1.5$ to $\approx1.49107$,
for  the special case when the item sizes are
in the range $(1/4,1/2]$, which corresponds to the $3$-partition case.
Extending the techniques for the $3$-partition case, Hebbar et al. \cite{soda-gadgets} have managed to break the barrier of $3/2$ on the upper bound of ARR of \bestfit{}
to achieve an upper bound of $3/2-\eps$ for some $\eps>10^{-10}$.
However, the conjecture that the ARR of \bestfit{} lies close to $1.15$ by Kenyon \cite{DBLP:conf/soda/Kenyon96} is wide open.

\section{Acknowledgements}
We thank Susanne Albers, Leon Ladewig, and Jir\'i Sgall for helpful initial discussions.
We thank Aditya Lonkar and Anish Hebbar for their inputs to simplify and improve some results.
We also thank several anonymous reviewers for their suggestions.
A part of this work was done when Nikhil Ayyadevara and Rajni Dabas were interns at the Indian Institute of Science.

Arindam Khan’s research is supported in part by Google India Research Award, SERB
Core Research Grant (CRG/2022/001176) on ``Optimization under Intractability and Uncertainty'',
and the Walmart Center for Tech Excellence at IISc (CSR Grant WMGT-23-0001).
K.~V.~N.~Sreenivas is grateful to the Google PhD Fellowship Program for supporting his research.
\bibliography{ref}

\newcommand{\etalchar}[1]{$^{#1}$}
\begin{thebibliography}{CJCG{\etalchar{+}}13}

\bibitem[AKL21a]{albers_et_al_MFCS}
Susanne Albers, Arindam Khan, and Leon Ladewig.
\newblock Best fit bin packing with random order revisited.
\newblock {\em Algorithmica}, 83(9):2833--2858, 2021.

\bibitem[AKL21b]{AlbersKL21}
Susanne Albers, Arindam Khan, and Leon Ladewig.
\newblock Improved online algorithms for knapsack and {GAP} in the random order
  model.
\newblock {\em Algorithmica}, 83(6):1750--1785, 2021.

\bibitem[BBD{\etalchar{+}}18]{BaloghBDEL18}
J{\'{a}}nos Balogh, J{\'{o}}zsef B{\'{e}}k{\'{e}}si, Gy{\"{o}}rgy D{\'{o}}sa,
  Leah Epstein, and Asaf Levin.
\newblock A new and improved algorithm for online bin packing.
\newblock In {\em ESA}, volume 112, pages 5:1--5:14, 2018.

\bibitem[BBD{\etalchar{+}}21]{BaloghBDEL19}
J{\'a}nos Balogh, J{\'o}zsef B{\'e}k{\'e}si, Gy{\"o}rgy D{\'o}sa, Leah Epstein,
  and Asaf Levin.
\newblock A new lower bound for classic online bin packing.
\newblock {\em Algorithmica}, 83(7):2047--2062, 2021.

\bibitem[BC81]{nfd}
Brenda~S Baker and Edward~G Coffman, Jr.
\newblock A tight asymptotic bound for next-fit-decreasing bin-packing.
\newblock {\em SIAM Journal on Algebraic Discrete Methods}, 2(2):147--152,
  1981.

\bibitem[BEK16]{BansalE016}
Nikhil Bansal, Marek Eli{\'{a}}s, and Arindam Khan.
\newblock Improved approximation for vector bin packing.
\newblock In {\em SODA}, pages 1561--1579, 2016.

\bibitem[BGK00]{bekesi20005}
J\'{o}zsef B\'{e}k\'{e}si, G\'{a}bor Galambos, and Hans Kellerer.
\newblock A 5/4 linear time bin packing algorithm.
\newblock {\em Journal of Computer and System Sciences}, 60(1):145--160, 2000.

\bibitem[BJL{\etalchar{+}}84]{bentley1984some}
Jon~Louis Bentley, David~S Johnson, Frank~Thomson Leighton, Catherine~C
  McGeoch, and Lyle~A McGeoch.
\newblock Some unexpected expected behavior results for bin packing.
\newblock In {\em STOC}, pages 279--288, 1984.

\bibitem[BLM13]{bernstein-ref}
Stéphane Boucheron, Gábor Lugosi, and Pascal Massart.
\newblock {\em Concentration Inequalities: A Nonasymptotic Theory of
  Independence}.
\newblock Oxford University Press, 2013.

\bibitem[BM15]{hoeffding-without-replacement}
R\'{e}mi Bardenet and Odalric-Ambrym Maillard.
\newblock Concentration inequalities for sampling without replacement.
\newblock {\em Bernoulli}, 21(3):1361--1385, 2015.

\bibitem[{Car}19]{fischer_thesis}
{Carsten Oliver Fischer}.
\newblock {\em New Results on the Probabilistic Analysis of Online Bin Packing
  and its Variants}.
\newblock PhD thesis, Rheinische Friedrich-Wilhelms-Universität Bonn, December
  2019.

\bibitem[CJCG{\etalchar{+}}13]{coffman2013bin}
Edward~G Coffman~Jr, J{\'{a}}nos Csirik, G{\'{a}}bor Galambos, Silvano
  Martello, and Daniele Vigo.
\newblock Bin packing approximation algorithms: survey and classification.
\newblock In {\em Handbook of combinatorial optimization}, pages 455--531.
  Springer New York, 2013.

\bibitem[CJJLS93]{coffman1993probabilistic}
Edward~G Coffman~Jr, David~S Johnson, George~S Lueker, and Peter~W Shor.
\newblock Probabilistic analysis of packing and related partitioning problems.
\newblock {\em Statistical Science}, 8(1):40--47, 1993.

\bibitem[CJJSW97]{coffman1997bin}
Edward~G Coffman~Jr, David~S Johnson, Peter~W Shor, and Richard~R Weber.
\newblock Bin packing with discrete item sizes, part ii: Tight bounds on first
  fit.
\newblock {\em Random Structures \& Algorithms}, 10(1-2):69--101, 1997.

\bibitem[CJK{\etalchar{+}}06]{csirik2006sum}
Janos Csirik, David~S Johnson, Claire Kenyon, James~B Orlin, Peter~W Shor, and
  Richard~R Weber.
\newblock On the sum-of-squares algorithm for bin packing.
\newblock {\em Journal of the ACM (JACM)}, 53(1):1--65, 2006.

\bibitem[CJSHY80]{coffman1980stochastic}
Edward~G Coffman~Jr, Kimming So, Micha Hofri, and AC~Yao.
\newblock A stochastic model of bin-packing.
\newblock {\em Information and Control}, 44(2):105--115, 1980.

\bibitem[DGV08]{dean2008approximating}
Brian~C Dean, Michel~X Goemans, and Jan Vondr{\'a}k.
\newblock Approximating the stochastic knapsack problem: The benefit of
  adaptivity.
\newblock {\em Mathematics of Operations Research}, 33(4):945--964, 2008.

\bibitem[dlVL81]{VegaL81}
W~Fernandez de~la Vega and George~S Lueker.
\newblock Bin packing can be solved within 1+epsilon in linear time.
\newblock {\em Combinatorica}, 1(4):349--355, 1981.

\bibitem[EPR13]{EisenbrandPR11}
Friedrich Eisenbrand, D\"{o}m\"{o}t\"{o}r P\'{a}lv\"{o}lgyi, and Thomas
  Rothvo\ss{}.
\newblock Bin packing via discrepancy of permutations.
\newblock {\em ACM Trans. Algorithms}, 9(3), 2013.

\bibitem[Fer89]{ferguson1989solved}
Thomas~S Ferguson.
\newblock Who solved the secretary problem?
\newblock {\em Statistical science}, 4(3):282--289, 1989.

\bibitem[FMMM09]{feldman2009online}
Jon Feldman, Aranyak Mehta, Vahab Mirrokni, and Shan Muthukrishnan.
\newblock Online stochastic matching: Beating 1-1/e.
\newblock In {\em FOCS}, pages 117--126, 2009.

\bibitem[FR16]{DBLP:conf/latin/FischerR16}
Carsten Fischer and Heiko R{\"{o}}glin.
\newblock Probabilistic analysis of the dual next-fit algorithm for bin
  covering.
\newblock In {\em LATIN}, pages 469--482, 2016.

\bibitem[FR18]{DBLP:conf/latin/FischerR18}
Carsten Fischer and Heiko R{\"{o}}glin.
\newblock Probabilistic analysis of online (class-constrained) bin packing and
  bin covering.
\newblock In {\em LATIN}, volume 10807, pages 461--474. Springer, 2018.

\bibitem[GKNS21]{gupta2021stochastic}
Anupam Gupta, Amit Kumar, Viswanath Nagarajan, and Xiangkun Shen.
\newblock Stochastic load balancing on unrelated machines.
\newblock {\em Mathematics of Operations Research}, 46(1):115--133, 2021.

\bibitem[GKR12]{gupta2012online}
Anupam Gupta, Ravishankar Krishnaswamy, and R~Ravi.
\newblock Online and stochastic survivable network design.
\newblock {\em SIAM Journal on Computing}, 41(6):1649--1672, 2012.

\bibitem[GS20]{Gupta020}
Anupam Gupta and Sahil Singla.
\newblock Random-order models.
\newblock In {\em Beyond the Worst-Case Analysis of Algorithms}, pages
  234--258. Cambridge University Press, 2020.

\bibitem[HKS24]{soda-gadgets}
Anish Hebbar, Arindam Khan, and K.~V.~N. Sreenivas.
\newblock Bin packing under random-order: Breaking the barrier of 3/2.
\newblock In {\em SODA}, pages 4177--4219, 2024.

\bibitem[HR17]{DBLP:conf/soda/HobergR17}
Rebecca Hoberg and Thomas Rothvoss.
\newblock A logarithmic additive integrality gap for bin packing.
\newblock In {\em SODA}, pages 2616--2625, 2017.

\bibitem[JCRZ08]{DBLP:journals/dam/CoffmanCRZ08}
Edward G~Coffman Jr, J{\'{a}}nos Csirik, Lajos R{\'{o}}nyai, and Ambrus
  Zsb{\'{a}}n.
\newblock Random-order bin packing.
\newblock {\em Discrete Applied Mathematics}, 156(14):2810--2816, 2008.

\bibitem[JDU{\etalchar{+}}74]{johnson1974worst}
David~S Johnson, Alan Demers, Jeffrey~D Ullman, Michael~R Garey, and Ronald~L
  Graham.
\newblock Worst-case performance bounds for simple one-dimensional packing
  algorithms.
\newblock {\em SIAM Journal on computing}, 3(4):299--325, 1974.

\bibitem[JG85]{mffd}
David~S Johnson and Michael~R Garey.
\newblock A 71/60 theorem for bin packing.
\newblock {\em Journal of Complexity}, 1(1):65--106, 1985.

\bibitem[Joh73]{johnson-thesis}
David~S Johnson.
\newblock {\em Near-optimal bin packing algorithms}.
\newblock PhD thesis, Massachusetts Institute of Technology, 1973.

\bibitem[Joh74]{DBLP:journals/jcss/Johnson74}
David~S Johnson.
\newblock Fast algorithms for bin packing.
\newblock {\em Journal of Computer and System Sciences}, 8(3):272--314, 1974.

\bibitem[Ken96]{DBLP:conf/soda/Kenyon96}
Claire Kenyon.
\newblock Best-fit bin-packing with random order.
\newblock In {\em SODA}, pages 359--364, 1996.

\bibitem[KK82]{KarmarkarK82}
Narendra Karmarkar and Richard~M Karp.
\newblock An efficient approximation scheme for the one-dimensional bin-packing
  problem.
\newblock In {\em FOCS}, pages 312--320, 1982.

\bibitem[KLMS84]{spaccamela}
Richard~M Karp, Michael Luby, and A~Marchetti-Spaccamela.
\newblock A probabilistic analysis of multidimensional bin packing problems.
\newblock In {\em STOC}, page 289–298, New York, NY, USA, 1984.

\bibitem[LL85]{lee-lee}
Chan~C Lee and Der-Tsai Lee.
\newblock A simple on-line bin-packing algorithm.
\newblock {\em Journal of ACM}, 32(3):562–572, July 1985.

\bibitem[LL87]{LL87}
Chan~C Lee and Der-Tsai Lee.
\newblock Robust on-line bin packing algorithms.
\newblock {\em Technical Report, Northwestern University}, 1987.

\bibitem[LS89]{leighton1989tight}
Tom Leighton and Peter Shor.
\newblock Tight bounds for minimax grid matching with applications to the
  average case analysis of algorithms.
\newblock {\em Combinatorica}, 9(2):161--187, 1989.

\bibitem[Mur88]{DBLP:journals/dam/Murgolo88}
Frank~D Murgolo.
\newblock Anomalous behavior in bin packing algorithms.
\newblock {\em Discrete Applied Mathematics}, 21(3):229--243, 1988.

\bibitem[MY11]{MahdianY11}
Mohammad Mahdian and Qiqi Yan.
\newblock Online bipartite matching with random arrivals: an approach based on
  strongly factor-revealing lps.
\newblock In {\em STOC}, pages 597--606, 2011.

\bibitem[NNN12]{NewmanNN12}
Alantha Newman, Ofer Neiman, and Aleksandar Nikolov.
\newblock Beck's three permutations conjecture: {A} counterexample and some
  consequences.
\newblock In {\em FOCS}, pages 253--262, 2012.

\bibitem[Ram89]{DBLP:journals/ipl/Ramanan89}
Prakash~V Ramanan.
\newblock Average-case analysis of the smart next fit algorithm.
\newblock {\em Information Processing Letters}, 31(5):221--225, 1989.

\bibitem[Rhe94]{rhee-ineq}
Wansoo~T Rhee.
\newblock Inequalities for bin packing-iii.
\newblock {\em Optimization}, 29(4):381--385, 1994.

\bibitem[RT87]{rhee-talagrand-martingales}
Wansoo~T Rhee and Michel Talagrand.
\newblock Martingale inequalities and np-complete problems.
\newblock {\em Mathematics of Operations Research}, 12(1):177--181, 1987.

\bibitem[RT88]{Rhee_Talagrand_Matching}
Wansoo~T Rhee and Michel Talagrand.
\newblock Exact bounds for the stochastic upward matching problem.
\newblock {\em Transactions of the American Mathematical Society},
  307(1):109--125, 1988.

\bibitem[RT93]{rhee1993lineB}
Wansoo~T Rhee and Michel Talagrand.
\newblock On-line bin packing of items of random sizes, ii.
\newblock {\em SIAM Journal on Computing}, 22(6):1251--1256, 1993.

\bibitem[Sho86]{shor1986average}
Peter~W Shor.
\newblock The average-case analysis of some on-line algorithms for bin packing.
\newblock {\em Combinatorica}, 6(2):179--200, 1986.

\bibitem[Spe94]{spencer1994ten}
Joel Spencer.
\newblock {\em Ten lectures on the probabilistic method}.
\newblock SIAM, 1994.

\end{thebibliography}
\appendix
\section{Greedy Algorithm for Maximum Upright Matching}
\label{app:upright-matching}
Here we show that \cref{alg:uprightmatching} correctly computes a maximum upright matching.
\begin{lemma}
    For a given set $P^+$ of $m$ plus points and a set $P^-$ of $m$ minus points, \cref{alg:uprightmatching} outputs a maximum upright matching.
\end{lemma}
\begin{proof}
For the purpose of the proof, we denote a valid upright matching of a given instance by a set of pairs $(p^+,p^-)$ where $p^+$ is a plus point
and $p^-$ is a minus point.
    Let $\opttt$ denote an optimal (maximum) matching
    and let $\algtt$ denote the matching output by \cref{alg:uprightmatching}
    and suppose that $\algtt\ne \opttt$.
    We will show that $\algtt$ is optimal via an exchange argument. In more detail, we show that we can modify $\opttt$ to obtain another optimal matching
    $\opttt'$ which is `closer' to $\algtt$. By repeating this modification process, we can reach $\algtt$ in a finite number of steps, thus showing that $\algtt$ is optimal.

    Without loss of generality, let us assume that all the points in the input instance have distinct $x$-coordinates and distinct $y$-coordinates. Further, for a point $p$, let $x(p)$ denote its $x$-coordinate and let $y(p)$ denote its $y$-coordinate. Define
    \begin{align*}
        \pi^+=\argmin_{p^+:(p^+,p^-)\in\algtt\triangle\opttt} x(p^+)
    \end{align*}
    where $\algtt\triangle\opttt$ denotes the symmetric difference between $\algtt$ and $\opttt$.
    In words, $\pi^+$ is the first plus point processed by \cref{alg:uprightmatching} that made $\opttt$ and $\algtt$ to differ.
    Now, we are ready to modify $\opttt$ and obtain $\opttt'$ without decreasing the number of matched pairs. We do so by considering three cases.

    \noindent\textbf{Case 1:} \emph{$\pi^+$ is matched in $\algtt$ but not in $\opttt$.}\\
    Say $\pi^+$ is matched to $p^-$ in $\algtt$. By the optimality of $\opttt$, the point $p^-$
    must be matched to some $p^+$ in $\opttt$. We can then obtain the new optimal matching $\opttt'=\opttt\cup\{(\pi^+,p^-)\}\setminus\{(p^+,p^-)\}$.

    \noindent\textbf{Case 2:} \emph{$\pi^+$ is matched in $\opttt$ but not in $\algtt$.}\\
    We show that this case cannot arise at all. Say $\pi^+$ is matched to $p^-$ in $\opttt$.
    In $\algtt$, $\pi^+$ is not matched, which means that $p^-$ is matched to some $p^+\ne \pi^+$ in $\algtt$.
    By the definition of $\pi^+$, it must be the case that $p^+$ is processed after $\pi^+$ by $\algtt$.
    But then $p^-$ must have been unmatched by the time $\pi^+$ is processed and \cref{alg:uprightmatching} would not have left
    $\pi^+$ unmatched.

    \noindent\textbf{Case 3:} \emph{$\pi^+$ is matched in both $\opttt$ and $\algtt$.}\\
    Say $\pi^+$ is matched with $p^-_1$ in $\algtt$ and with $p^-_2$ in $\opttt$.
    If $p^-_1$ is unmatched in $\opttt$, then we can define a new optimal matching $\opttt'=\opttt\cup\{(\pi^+,p_1^-)\}\setminus\{(\pi^+,p_2^-)\}$.
    Now, suppose $p^-_1$ is matched to $p^+$ in $\opttt$. Then we claim that $y(p^-_1)\ge y(p^-_2)$. This is because at the time $\pi^+$ was considered by
    \cref{alg:uprightmatching}, both $p_1^-$ and $p_2^-$ must have been unmatched. (Otherwise, it violates the definition of $\pi^+$.) Since
    \cref{alg:uprightmatching} preferred $p_1^-$ over $p_2^-$, it must be the case that $y(p^-_1)\ge y(p^-_2)$. Hence $(p^+,p_2^-)$ is a valid matching pair.
    Thus we can obtain the new optimal matching $\opttt'=\opttt\cup\{(\pi^+,p_1^-),(p^+,p_2^-)\}\setminus \{(\pi^+,p_2^-),(p^+,p_1^-)\}$.

    That ends the construction of $\opttt'$. To conclude the proof, we define 
    \begin{align*}
    x^+(\mathcal M_1,\mathcal M_2)=\min_{(p^+,p^-)\in\mathcal M_1\triangle\mathcal M_2}x(p^+)
    \end{align*}
    for any two matchings $\mathcal M_1,\mathcal M_2$, and observe that
    $x^+(\opttt',\algtt)$ is strictly greater than $x^+(\opttt,\algtt)$.
\end{proof}
\section{Optimal Packing Size is Almost Proportional to Input Length}
\label{app:proportionality}
In this section, we show \cref{claim:rhee-talagrand}.
First, we use the following lemma, which is almost similar to Theorem 2.1 in \cite{rhee1993lineB}.
\begin{lemma}
    \label{lem:rhee-talagrand-without-replacement}
    Let $J$ be an arbitrary set of $q$ items and let $S$ be a set of $p$ items (with $p\le q$) sampled uniformly randomly \emph{without} replacement from $J$. Then there exist constants $K',a'$
    such that
    \begin{align*}
        \prob{\opt(S)\ge\frac{p}{q}\opt(J)+K'\sqrt q(\log q)^{3/4}}\le \exp\left(-a'(\log q)^{3/2}\right).
    \end{align*}
\end{lemma}
Let $\Isamp(t)$ denote a set of $t$ items sampled uniformly randomly \emph{without} replacement from $I$. Then, using the above lemma, we obtain that
\begin{align}
    \prob{\opt(\Isamp(t))\ge\frac{t}{n}\opt(I)+K'\sqrt n(\log n)^{3/4}}\le \exp\left(-a'(\log n)^{3/2}\right).\label{eq:rhee-talagrand-without-replacement}
\end{align}
Then, we observe that the lists $I(1,t)$ and $\Isamp(t)$ have the same joint distribution. This implies that both the quantities
$\opt(I(1,t))$ and $\opt(\Isamp(t))$ are roughly the same.
\begin{proposition}
    \label{prop:first-t-and-without-replacement}
    There exist constants $K'',a''$ such that
    \begin{align*}
        \prob{\opt(I(1,t))-\opt(\Isamp(t))\ge K''\sqrt n(\log n)^{3/4}}\le \exp\left(-a''(\log n)^{3/2}\right).
    \end{align*}
\end{proposition}
To conclude, we observe that combining \cref{eq:rhee-talagrand-without-replacement} and the above proposition gives us the final claim.
We are left with proving \cref{lem:rhee-talagrand-without-replacement} and \cref{prop:first-t-and-without-replacement}.

\begin{proof}[Proof of \cref{lem:rhee-talagrand-without-replacement}]
    The proof is long, but is essentially the same as that of Theorem 2.1 in \cite{rhee1993lineB}. The only difference is 
    the use of concentration inequalities for sampling without replacement, instead of sampling with replacement.
    
    Let $(Y_1,Y_2,\dots,Y_q)$ denote the set $J$ and let $s(1),s(2),\dots,s(p)$ denote $p$ indices that are sampled uniformly
    randomly from $[q]$ without replacement. Hence $(Y_{s(1)},Y_{s(2)},\dots,Y_{s(p)})$ denotes the set $S$.
    We first compute an optimal packing of $J$. We refer to this packing as \emph{model packing} and denote it by $\mathcal P$.
    Without loss of generality, we assume that in each bin of $\mathcal P$, the items are ordered in non-increasing order of weights.
    For an item $j\in J$, we define $\rank(j)$ as the position of $j$ in the bin in which $j$ is placed in $\mathcal P$.

    Now, we show how to pack $S$. For every rank $r\in\{2,3,\dots,p\}$, we maintain a set $V(r)$ of pairs of the form $(B,x)$.
    Initially, all these sets are empty.
    We process $Y_{s(1)},Y_{s(2)},\dots,Y_{s(p)}$ in that order.
    Some of these items will be tagged as \emph{overflow items}. All the overflow items will be packed at once using Next-Fit.
    Assume we are at an intermediate stage where we have packed $Y_{s(1)},\dots,Y_{s(i-1)}$ and are about to pack $Y_{s(i)}$.
    \begin{enumerate}
        \item If $\rank(Y_{s(i)})=1$, then we open a new bin $B$ and pack $Y_{s(i)}$ there.
        In addition, for each rank $r\in\{2,3,\dots,m_i\}$ (where $m_i$ is the number of items in the bin in $\mathcal P$ containing $Y_{s(i)}$), we add the pair $(B,x)$ to the list $V(r)$, where $x$
        is the weight of the $r\Th$ largest item in the bin in $\mathcal P$ containing $Y_{s(i)}$.
        \item If $r\coloneqq \rank(Y_{s(i)})\in\{2,\dots,p\}$, we look at all the pairs in $V(r)$ and identify the pair $(B,x)$, if one exists,
        such that $x$ is least possible but at least the weight of $Y_{s(i)}$. We then pack $Y_{s(i)}$ in $B$ and remove the pair $(B,x)$ from
        the list $V(r)$. If no such pair in $V(r)$ exists, then we tag $Y_{s(i)}$ as a overflow item.
        \item If $r\coloneqq \rank(Y_{s(i)})>p$, we tag $Y_{s(i)}$ as a overflow item.
    \end{enumerate}
    In this way, we pack all the items in $S$.
    We now analyze the number of bins used in this process.
    First, observe that, ignoring the overflow items, we open a new bin only when an item of rank $1$ appears.
    Let $N_{\main}$ denote this number of bins.
    Since $S$ is obtained by sampling $p$ items from $J$ without replacement,
    by using Hoeffding's inequality for sampling without replacement (see \cite{hoeffding-without-replacement}),
    we can bound $N_{\main}$.

    We state the Hoeffding's inequality now.
    \begin{proposition}[Hoeffding's Inequality]
        Let $(Z_1,Z_2,\dots,Z_M)$ be a list of $M$ arbitrary real numbers in $[0,1]$ and let $\mu=\frac 1M\sum_j Z_j$ be the mean of this list. Let $(z_1,z_2,\dots,z_m)$ be a list of $m$ numbers sampled from the original list uniformly at random without replacement. Then for any $\delta>0$,
        \begin{align*}
            \prob{\sum_{i\in[m]}z_i-\mu m>\delta}\le \exp\left(-\frac{2\delta^2}{m}\right).
        \end{align*}
    \end{proposition}
To bound $N_{\main}$, we use the Hoeffding's inequality by setting
\begin{enumerate}
    \item $M=q$ and $m=p$.
    \item $Z_j=1$ if $\rank(Y_j)=1$ and $Z_j=0$ otherwise.
\end{enumerate}
Then, we obtain that $\mu=1/q\sum_{j}Z_j=\opt(J)/q$ and 
$\sum_{i}z_i=N_{\main}$.
Hence we obtain
    \begin{align*}
        \prob{N_{\main}\ge \frac pq \opt(J)+ \sqrt p(\log q)^{3/4}}\le \exp(-2(\log q)^{3/2}).
    \end{align*}
    We are left with bounding the number of bins to pack the overflow items. We denote the total weight by $W_{\of}$ and the number of bins used to pack them by $N_{\of}$.
    Since we use Next-Fit to pack the overflow items, we have $N_{\of}\le 2W_{\of}+1$. Hence, we only need to bound $W_{\of}$. Note that
    any item of rank $r$ can have weight at most $1/r$. Hence the total weight of items of rank more than $p$ (those considered in step 3 of our procedure) is at most $1$.
    This leaves us with the task of bounding the total weight of overflow items of rank in $\{2,3,\dots,p\}$ (step 2 in our procedure).

    For a rank $r\in\{2,\dots,p\}$, let $n_r$ denote the number of rank $r$ overflow items. We can bound $n_r$ by reducing it to a variant of upright matching.
    For each $r\in\{2,\dots,p\}$, we gradually construct an instance of upright matching $\mathcal M_r$ as we process $S$ as follows. Suppose we are processing $Y_{s(i)}$.
    \begin{enumerate}
        \item Say $\rank(Y_{s(i)})=1$ and there are at least $r$ items in the bin in $\mathcal P$ containing $Y_{s(i)}$.
        Then we create a minus point $(i,w)$ where $w$ is the weight of the rank $r$ item in $\mathcal P$ containing $Y_{s(i)}$.
        \item Say $\rank(Y_{s(i)})=r$. Then we create a plus point $(i,w)$ where $w$ is the weight of $Y_{s(i)}$.
        \item We do nothing in rest of the cases.
    \end{enumerate}
    Then, it can be seen that the maximum matching procedure of \cref{alg:uprightmatching} on $\mathcal M_r$ exactly corresponds to step 2 of our packing procedure.\footnote{except for reflection about $x$-axis.} Hence the number of unmatched plus points in $\mathcal M_r$ exactly corresponds to the number of overflow items of rank $r$.

    For the sake of simplicity, and to overcome certain technical details of upright 
    matching, we assume that instead of stopping after sampling $p$ items from $J$, we continue till all the $q$ items of $J$ are sampled.
    This will only increase the number of overflow items.

    Let us fix a rank $r\in\{2,\dots,p\}$. Let $q_r$ denote the total number of items of rank $r$ in $J$.
    From the results of \cite{fischer_thesis} (see \textsc{Matching Variant} M2 in the reference),
    the number of unmatched plus points in $\mathcal M_r$ can be bounded as $n_r\le K\sqrt{q_r}(\log q_r)^{3/4}$ with probability at least
    $1-\exp(-a(\log q_r)^{3/2})$ for some constant $a>3$. Distinguishing between the cases $q_r\le q^{1/3}$ and $q_r\ge q^{1/3}$, we obtain that
    \begin{align*}
        n_r\le q^{1/3}+K\sqrt{q_r}(\log q_r)^{3/4},
    \end{align*}
    with probability at least $1-\exp(-\frac a3(\log q)^{3/2})$.
    Now applying a union bound, we obtain that for all $r\in\{2,\dots,p\}$, we have
    \begin{align*}
        n_r\le q^{1/3}+K\sqrt{q_r}(\log q)^{3/4}
    \end{align*}
    with probability at least $1-q\exp(-\frac a3(\log q)^{3/2})$ which is at least $1-\exp(-(\frac a3-1)(\log q)^{3/2})$.

    Hence we obtain that the total number of overflow items is at most
    \begin{align*}
        \sum_{r=2}^p \frac{n_r}{r}\le q^{1/3}\log q+K(\log q)^{3/4}\sum_{r=2}^{p}\frac{\sqrt{q_r}}{r}.
    \end{align*}
    Applying Cauchy--Schwartz inequality and using the fact that $\sum_{r=2}^p{q_r}\le q$, we obtain that
    \begin{align*}
        \sum_{r=2}^{p}\frac{\sqrt{q_r}}{r}\le \sqrt{\sum_{r=2}^p q_r}\sqrt{\sum_{r=2}^p \frac{1}{r^2}}\le K\sqrt q,
    \end{align*}
    thus ending the proof.
\end{proof}

\begin{proof}[Proof of \cref{prop:first-t-and-without-replacement}]
    Let $X_1,X_2,\dots,X_n$ be the items in the input set $I$. Let $s(1),s(2),\dots,s(t)$
    denote $t$ indices from $[n]$ sampled uniformly randomly without replacement.
    Therefore, $I(1,t)=(X_1,X_2,\dots,X_t)$ and $\Isamp(t)=(X_{s(1)}, X_{s(2)},\dots,X_{s(t)})$.
    We now show that $I(1,t)$ and $\Isamp(t)$ have the same joint distribution by showing that the joint cumulative distribution functions are equal.
    Let $x_1,x_2,\dots,x_t$ be some arbitrary reals in $[0,1]$. Then
    \begin{align*}
        &\quad\prob{\left(X_{s(1)}\le x_1\right)\land \left(X_{s(2)}\le x_2\right)\land \dots\land \left(X_{s(t)}\le x_t\right)}\\
            &=\sum_{\substack{s_1,s_2,\dots,s_t\in[n]\\\text{all different}}}\prob{\bigwedge_{i\in [t]} X_{s(i)}\le x_i{\Bigg|} \bigwedge_{i\in [t]} s(i)=s_i}\prob{\bigwedge_{i\in [t]} s(i)=s_i}\\
            &=\sum_{\substack{s_1,s_2,\dots,s_t\in[n]\\\text{all different}}}\prob{\bigwedge_{i\in [t]} X_{s_i}\le x_i}\prob{\bigwedge_{i\in [t]} s(i)=s_i}\\
            &=\sum_{\substack{s_1,s_2,\dots,s_t\in[n]\\\text{all different}}}\prob{\bigwedge_{i\in [t]} X_{i}\le x_i}\prob{\bigwedge_{i\in [t]} s(i)=s_i}\\
            &=\prob{\bigwedge_{i\in [t]} X_{i}\le x_i}\sum_{\substack{s_1,s_2,\dots,s_t\in[n]\\\text{all different}}}\prob{\bigwedge_{i\in [t]} s(i)=s_i}\\
            &=\prob{(X_{1}\le x_1)\land (X_{1}\le x_2)\land \dots\land (X_{t}\le x_t)}.
    \end{align*}
    Since $I(1,t)$ and $\Isamp(t)$ have the same joint distribution, we obtain that
    \begin{align*}
        \expec{\opt(I(1,t))}=\expec{\opt(\Isamp(t))}\eqqcolon\Gamma.
    \end{align*}
    Further, Theorem 2 from \cite{rhee-talagrand-martingales} tells us that if a set $T$ of $m$ items is sampled independently and identically from a fixed distribution, then for any $\ell\ge 0$,
    \begin{align*}
        \prob{\abs{\opt(T)-\expec{\opt(T)}}>\ell}\le 2\exp(-\ell^2/(2m)).
    \end{align*}
    Thus, we obtain the following two inequalities, which hold for some constants $K,a$.
    \begin{align*}
        \prob{\abs{\opt(I(1,t))-\Gamma}>K\sqrt t(\log n)^{3/4}}\le 2\exp(-a(\log n)^{3/2}),\\
        \prob{\abs{\opt(\Isamp(t))-\Gamma}>K\sqrt t(\log n)^{3/4}}\le 2\exp(-a(\log n)^{3/2}).
    \end{align*}
    Combining both the inequalities and observing that $t\le n$ gives us the proposition.
\end{proof}

\end{document}